\documentclass[11pt]{article}

\usepackage[margin=1in]{geometry}

\usepackage{amsmath,amsthm,amsfonts,bm,bbm} 

\usepackage[table]{xcolor}
\usepackage{mathtools}
\usepackage{graphicx}
\usepackage{array}
\usepackage{arydshln}
\usepackage{diagbox}
\usepackage{multirow}
\usepackage{enumitem}
\usepackage{complexity}
\usepackage{colortbl}

\usepackage{hyperref}
\usepackage[capitalize,noabbrev]{cleveref}

\newtheorem{theorem}{Theorem}
\newtheorem{lemma}{Lemma}
\newtheorem{corollary}{Corollary}

\theoremstyle{definition}
\newtheorem{definition}{Definition}[section]

\newcommand{\bR}{\mathbb{R}}
\newcommand{\bZ}{\mathbb{Z}}

\newcommand{\supp}{\textrm{supp}}

\DeclareMathOperator*{\argmax}{arg\,max}

\newcommand{\RESBI}[1]{\textsc{ReBimatrix}$(#1)$}
\newcommand{\PURE}{\textsc{PureCircuit}}
\newcommand{\POLY}{\textsc{RePolymatrix}}
\newcommand{\NOT}{\textsc{Not}}
\newcommand{\AND}{\textsc{And}}
\newcommand{\PURIFY}{\textsc{Purify}}

\newcommand{\red}[1]{\textcolor{red}{#1}}
\newcommand{\blue}[1]{\textcolor{blue}{#1}}

\newcommand{\itemwithtable}[2]{%
\item[]%
\parbox[t]{\linewidth}{%
  \begin{minipage}[t]{0.55\textwidth}\vspace{0pt}#1\end{minipage}\hfill
  \begin{minipage}[t]{0.4\textwidth}\vspace{0pt}\centering #2\end{minipage}%
}%
}

\usepackage{amssymb}
\usepackage{booktabs}
\usepackage[ruled]{algorithm2e}
\usepackage{nicematrix}

\usepackage[square,numbers]{natbib}
\SetAlFnt{\small}
\SetAlCapFnt{\small}
\SetAlCapNameFnt{\small}
\SetAlCapHSkip{0pt}
\IncMargin{-\parindent}

\newcommand{\eps}{\ensuremath{\varepsilon}}
\DeclareMathOperator{\encode}{enc}
\DeclareMathOperator{\cb}{CB}
\DeclareMathOperator{\rb}{RB}
\newcommand{\reals}{\mathbb{R}}

\title{The Complexity of Sparse Win-Lose Bimatrix Games}

\author{
  Eleni Batziou$^{1}$,
  John Fearnley$^{1}$,
  Abheek Ghosh$^{2}$, and
  Rahul Savani$^{1,3}$ \\
  \small
  $^{1}$University of Liverpool, \texttt{\{eleni.batziou, john.fearnley, rahul.savani\}@liverpool.ac.uk} \\
  \small
  $^{2}$Technical University of Munich, \texttt{abheek.ghosh@tum.de} \\
  \small
  $^{3}$The Alan Turing Institute
}

\date{}  

\begin{document}

\maketitle

\begin{abstract}
\noindent
We prove that computing an $\eps$-approximate Nash equilibrium of a win-lose bimatrix game with \emph{constant} sparsity is PPAD-hard for inverse-polynomial $\eps$.
Our result holds for 3-sparse games, which is tight given that 2-sparse win-lose bimatrix games can be solved in polynomial time.
\end{abstract}


\begin{titlepage}

\maketitle

\vspace{1cm}
\setcounter{tocdepth}{2} 
\tableofcontents

\end{titlepage}


\section{Introduction}

An important strand of research within algorithmic game theory has studied the computational complexity
of finding an (approximate) Nash equilibrium in different classes of games.
One of the most basic models of a game is a \emph{bimatrix game}, namely a two-player game in strategic form,
and a celebrated result showed that it is PPAD-complete to find an $\epsilon$-Nash equilibrium
of a bimatrix game for $\eps$ that is inverse polynomial in
the input size of the bimatrix game~\cite{DaskalakisGP09,ChenDT09}.
PPAD-hardness is generally considered as strong evidence that a problem cannot be solved in polynomial time~\cite{Goldberg2011}.

Around the same time as those breakthrough results, it was shown that this result extends to \emph{win-lose}
bimatrix games, where all payoffs are in $\{0, 1\}$.
In particular, this was shown via a polynomial-time reduction from the problem of finding
an approximate Nash equilibrium of a bimatrix game with general payoffs to that of finding an
approximate equilibrium of a win-lose bimatrix game~\cite{abbott2005complexity,chen2007approximation}.%
\footnote{The results in~\cite{abbott2005complexity} were for \emph{exact} Nash equilibria;
\citet{chen2007approximation} then showed that the \emph{same} construction actually
gives PPAD-hardness for $\eps$-Nash equilibria with inverse polynomial $\eps$.}
The win-lose games arising from the corresponding reduction are not \emph{sparse} in general, meaning
that rows and columns can contain many 1s as well as 0s.
In a separate line of work, it was shown to be PPAD-hard to find an
$\epsilon$-Nash equilibrium of a 10-sparse (but \emph{not} win-lose) bimatrix game -- i.e.,
a bimatrix game with at most 10 non-zero entries per row and per column -- again
for $\eps$ that is inverse polynomial in the input size of the bimatrix game~\cite{ChenDT06}.

\citet{ChenDT06} and, contemporaneously, \citet{CodenottiLR06} showed
that 2-sparse win-lose games can be solved in polynomial time.
This left open the possibility that less sparse win-lose bimatrix games could also
be solved in polynomial time.
Progress was made some years later, when \citet{liu2018approximation} showed
PPAD-hardness for inverse polynomial approximations in win-lose bimatrix games with $\log(n)$-sparsity.
Whether \emph{constant}-sparse win-lose bimatrix games can be solved in polynomial time for inverse-polynomial approximations remained an open question.
However, constant-sparse bimatrix games admit a PTAS \citep{DaskalakisP09,Barman2015}, so constant approximations can be computed in polynomial time.

No further progress on this question has been made for bimatrix games, but
results have been shown for the more general class of \emph{polymatrix games}.
\citet{LiuLD21} showed that win-lose polymatrix games with constant sparsity are PPAD-hard.
They also explicitly noted that their techniques for showing hardness in polymatrix games would not extend to
bimatrix games and left the problem of trying to show PPAD-hardness for win-lose bimatrix
games with constant sparsity as the main open problem in this line of work.
More recently, \citet{deligkas2024pure} also showed a strong hardness result
for win-lose \emph{non-sparse polymatrix} games, which we will utilise in this
paper.

\paragraph{\bf Our contribution.}

In this paper we fully resolve the complexity of finding an approximate
equilibrium of a win-lose bimatrix game with constant sparsity.
We show that
it is PPAD-hard to find an $\eps$-Nash equilibrium of a 3-sparse win-lose bimatrix game for inverse polynomial $\eps$.
Recall that 2-sparse win-lose bimatrix games can be solved in polynomial time,
which means that our result is tight: 2-sparse games are easy, whereas 3-sparse
games are PPAD-hard.

Our reduction starts from the \PURE{} circuit problem, which is known to be
PPAD-complete~\cite{deligkas2024pure}. It then proceeds via polymatrix games,
and then on to $3$-sparse $\{0, 1, 2, 8\}$-valued bimatrix games. It then
uses a sequence of steps in which the maximum payoff of the game is
reduced while maintaining constant sparsity, finally giving a
$3$-sparse $\{0, 1\}$-valued bimatrix game.

While several of these steps are already known, such as a reduction from
\PURE{} to polymatrix games, or a reduction from polymatrix games to
bimatrix games, at each step we need to carefully refine those reductions in
order to control the sparsity of the final game.

\begin{enumerate}
\item While it is already known that polymatrix games with payoffs in $\{0,
1\}$ are PPAD-hard~\cite{deligkas2024pure}, we cannot start with those games.
This is because in those games each player's payoff can be the sum of up to six
$1$s, and when we reduce such a game to a bimatrix game, we would end up with a game with sparsity at least $6$, which is too large for our purposes.\footnote{We were able to construct an optimized reduction from \PURE{} to polymatrix games with $\{0, 1\}$-payoffs where each player's payoff can be the sum of up to three $1$s. However, when we have $m$ many $1$s per player in the polymatrix game, we end up having sparsity of at least $m + 1$ in the bimatrix game. So, with six (three) $1$s per player in the polymatrix game, we end up having sparsity of seven (four) in the bimatrix game, which is too large for our purposes.}

Instead, we find that it is more convenient to work with polymatrix games that
use payoffs in the set $\{0, 1, 2\}$, which allows the sparsity of the bimatrix
game to be reduced. We carefully construct our hard polymatrix
instances so that the resulting bimatrix game has sparsity $3$.

\item We then reduce the polymatrix game to a bimatrix game. Here we use a
polymatrix game to bimatrix game reduction that is different from the usual
procedures in the literature~\cite{DaskalakisGP09}.
The previous reductions produce one or more of the following in the constructed bimatrix game:
negative payoffs, non-constant payoffs, or non-constant sparsity.
Our technique resolves these issues and produces a $3$-sparse $\{0, 1, 2, 8\}$-valued bimatrix game.

\item The final task is to reduce the maximum payoff in the bimatrix game
without blowing up the sparsity. While there are generic techniques in the
literature for constructing win-lose
games~\cite{abbott2005complexity,chen2007approximation} from arbitrary bimatrix
games, they use too many payoffs to maintain our desired sparsity, and so they
are not useful for our purposes.

Instead, we introduce a small set of bimatrix game to bimatrix game reductions
that slowly reduce the maximum payoff of the game while keeping a handle on the
sparsity. This allows us to produce a
$3$-sparse $\{0, 1\}$-valued bimatrix game, thereby proving our main result.
\end{enumerate}
We also show that each of the steps above hold even for inverse
polynomial approximate equilibria; this is unlikely to be improvable since
 a constant additive approximation can be computed in quasi-polynomial time~\cite{LiptonMM03,Rubinstein16}.
We give more details about our techniques in Section~\ref{sec:to}.


\section{Related Work}
The seminal results of \citet{DaskalakisGP09} and \citet{ChenDT09} established that computing Nash equilibria in two-player (bimatrix) games is PPAD-complete. \citet{Goldberg2011} provides an extensive survey that explores the limits of tractability for equilibrium computation and highlights restricted subclasses of games where efficient algorithms are possible.

An extensive line of work has studied restrictions on the payoff values of bimatrix games. \citet{abbott2005complexity} show that restricting payoff values to $\{0,1\}$ (win-lose games) does not allow for polynomial-time algorithms in the general case and the equilibrium computation problem remains PPAD-hard. \citet{chen2007approximation} study approximation in win-lose games and subsequently show PPAD-hardness of finding an $\eps$-approximate Nash equilibrium for $\eps = 1/\poly(n)$.

Restricting their focus to sparse bimatrix games, \citet{ChenDT06} show that sparsity does not
guarantee tractability, establishing PPAD-hardness for games where payoff
matrices have at most 10 non-zero entries per row and column. Closest to our setting are results on games that are both win-lose and sparse. \citet{CodenottiLR06}
give an efficient algorithm for the problem of finding a Nash equilibrium in
2-sparse win-lose games. \citet{liu2018approximation} establish PPAD-hardness for
win-lose bimatrix games with $\log (n)$-sparsity, and \citet{LiuLD21} extend
the result to multi-player games. Both works identify the
constant-sparsity regime as the main open question, which we resolve in this
paper.

The question of tractability in win-lose games has inspired a line of work studying restrictions on the adjacency graph that can lead to efficient algorithms. Games with two players can be represented as directed bipartite graphs. \citet{Addario007} study win-lose bimatrix games under the restriction of planarity and show that such games always have a uniform Nash equilibrium that can be found by a polynomial-time algorithm. Their work poses the question of how far tractability extends to other graph classes, and our work contributes by providing a hardness result for bounded degree instances derived by 3-sparsity. In follow-up work, \citet{Datta2011} strengthen the previous result by providing a logarithmic-space algorithm and extend tractability to additional graph classes defined by minor restrictions.

On the algorithmic side, for constant $\eps$, \citet{Kontogiannis2007} propose polynomial-time algorithms to compute well-supported approximate Nash equilibria in win-lose bimatrix games using graph-theoretic methods.
\citet{Barman2015} gives a polynomial-time approximation scheme (PTAS) for bimatrix games in which the sum of the payoff matrices has fixed column sparsity, using an approximate version of Carathéodory's theorem.
\citet{Collevecchio2025} analyze random win-lose games, with payoff entries drawn independently from a Bernoulli distribution, and obtain expected polynomial-time algorithms for broad parameter regimes.

\citet{Hermelin2011} study the parametrized complexity of computing Nash equilibria in bimatrix games and show fixed-parameter tractability under the joint restriction of constant sparsity and support size, namely the set of strategies played with non-zero probability. \citet{Bilo2023} consider the decision version of the problem in multi-player games and show that hardness persists even under the win-lose assumption beyond the two player setting. 

The hardness result of \citet{abbott2005complexity} has been used as a starting point for reductions in various settings. \citet{Papadimitriou2023} prove PPAD-hardness for public goods games on directed networks, obtaining hardness for divisible goods with summation utilities via a reduction from the problem of approximating mixed Nash equilibria in win-lose bimatrix games. \citet{Brandt2009} study ranking games and show that any $k$-player game is equivalent to a $k+1$-zero-sum game. They prove PPAD-hardness for computing mixed Nash equilibria in 3-player ranking games via a reduction from win-lose 2-player games using Nash homomorphisms, illustrating how hardness of win-lose bimatrix games can be extended to structured multi-player settings. By showing hardness for \emph{3-sparse} win-lose games, we open up the possibility of starting with
these games for hardness reductions which may lead to other such results.

\section{Preliminaries}
Let $[n]$ denote the set $\{0, 1, 2, \ldots, n-1\}$ and let $\Delta(n)$ denote the $n-1$ dimensional simplex. Let $\bmod(k, n) \in [n]$ be the modulo of integer $k \in \bZ$ with respect to positive integer $n \in \bZ_{>0}$.

In a bimatrix game, there are two players, the row player and the column player, with $n$ and $m$ actions, respectively. We represent the game as $(A, B)$, where $A$ and $B$ are $n \times m$ payoff matrices of the row and the column players, respectively.
A pure strategy of the row player is a row $i \in [n]$. A mixed strategy of the row player is a distribution over the rows $x \in \Delta(n)$. Similarly, pure and mixed strategies of the column player are of the form $j \in [m]$ and $y \in \Delta(m)$. For pure strategies $i \in [n]$ and $j \in [m]$ the payoff pair is $(A_{ij}, B_{ij})$.
Given a mixed profile $(x, y)$ the expected payoffs to the row and column players are $x^\intercal A y$ and $x^\intercal B y$, respectively. The \emph{support} of a distribution $x$ is denoted $\supp(x)=\{ i \in [n] : x_i > 0 \}$.

\paragraph{Nash equilibrium and approximations.}
A strategy profile $(x, y) \in \Delta(n) \times \Delta(m)$ is a (mixed-strategy) Nash equilibrium (NE) if neither player can gain by unilaterally deviating:
\[
    x^\intercal A y \ge x'^\intercal A y \quad\text{for all } x' \in \Delta(n),
    \qquad \text{and} \qquad
    x^\intercal B y \ge x^\intercal B y' \quad\text{for all } y' \in \Delta(m).
\]

For $\eps \ge 0$, we say $(x, y)$ is an $\eps$-Nash equilibrium ($\eps$-NE) if no player can improve her expected payoff by more than $\eps$ by a unilateral deviation. Formally,
\[
    x^\intercal A y + \eps \ge x'^\intercal A y \quad\text{for all } x' \in \Delta(n),
    \qquad \text{and} \qquad
    x^\intercal B y + \eps \ge x^\intercal B y' \quad\text{for all } y' \in \Delta(m).
\]

A stronger notion is that of an $\eps$-well-supported Nash equilibrium ($\eps$-WSNE): here every pure strategy that is played with positive probability must be an $\eps$-best-response. Formally, $(x, y)$ is an $\eps$-WSNE if
\begin{align*}
    e_i^\intercal A y + \eps &\ge x'^\intercal A y, \quad\text{for all } i \in \supp(x), x' \in \Delta(n),\\
    x^\intercal A e_j + \eps &\ge x^\intercal A y', \quad\text{for all } j \in \supp(y), y' \in \Delta(m),
\end{align*}
where $e_i$ denotes the $i$-th standard basis vector.


Note that an $\eps$-WSNE is also an $\eps$-NE, but the converse need not hold.
However, an approximate NE can be converted to an approximate WSNE with a
polynomial loss in error.

\begin{lemma}{\cite[Lemma 3.2]{ChenDT09}}
\label{lm:ne2wsne}
In a bimatrix game $(A, B)$ with $A, B \in [0, 1]^{n \times m}$, given
any $(\eps^2/8)$-NE, we can find an $\eps$-WSNE in polynomial time.
\end{lemma}

\paragraph{Sparsity.}
A matrix $M$ is called $k$-row-sparse if every row of $M$ has at most $k$ non-zero entries and $k$-column-sparse if every column of $A$ has at most $k$ non-zero entries. If $A$ is both $k$-row-sparse and $k$-column-sparse, we simply call it $k$-sparse. The bimatrix game $(A, B)$ is called $k$-sparse if both payoff matrices $A$ and $B$ are $k$-sparse.

\paragraph{Win--lose and $S$-valued games.}
A bimatrix game $(A,B)$ is a win--lose or $\{ 0, 1 \}$-valued game if all entries of $A$ and $B$ belong to $\{ 0, 1 \}$. More generally, we say a game is an $S$-valued game when all entries lie in a set $S \subseteq \bR$.


\section{Technical Overview}
\label{sec:to}

Our goal is to show that computing an $\eps$-NE of $3$-sparse $\{ 0, 1
\}$-valued bimatrix game is PPAD-hard for inverse polynomial $\eps$. In this
section we give a high-level sketch for showing PPAD-hardness of exact Nash
equilibria. In later sections we formally prove that these techniques can be
extended for approximate equilibria with inverse polynomial $\eps$.

We start the reduction from the \PURE{} problem. \PURE{} is a PPAD-hard
fixed-point computation problem~\cite{deligkas2024pure}. We reduce this problem to the equilibrium
computation problem in a class of \textit{polymatrix} games with the following
properties (see \cref{sec:res-bimatrix} for a formal definition):
(i) every player has two actions,
(ii) the game uses payoffs only from the set $\{ 0, 1, 2 \}$,
(iii) all payoff matrices are diagonal or anti-diagonal,
(iv) the underlying graph of the polymatrix game is bipartite, and
(v) every player can affect the payoff or has a payoff that depends upon at most two other players.
These properties become useful in our subsequent reduction steps.
\citet{deligkas2024pure} also give a similar reduction from \PURE{} to
polymatrix games, but their reduction does not ensure that all of the properties
above hold. We need these properties for our further reductions, so we show
that their reduction can be adapted to obtain these properties.

The two main technical contributions of the paper are to show that the class of
polymatrix games described above can be reduced to $3$-sparse $\{ 0, 1, 2, 8
\}$-valued bimatrix games, and that the $3$-sparse $\{ 0, 1, 2, 8 \}$-valued
bimatrix games can be reduced to $3$-sparse $\{ 0, 1 \}$-valued bimatrix games.
The next two subsections provide an overview of these two steps.

\subsection{$\{ 0, 1, 2, 8 \}$-Valued Bimatrix Games}
We reduce the class of polymatrix games described above to $3$-sparse $\{ 0, 1, 2, 8 \}$-valued bimatrix games (with some additional properties).

Recall that the input polymatrix game is assumed to be bipartite.
Without loss of generality we can assume that the two sides of the bipartite polymatrix game have equal number of players, say $n$.
We construct a bimatrix game with $3n + 2$ actions per player, where the row (column) player simulates the actions of the $n$ polymatrix players on the left (right) side of the bipartite graph.

We call the first $2n$ actions per bimatrix player \textit{primary} actions.
These actions correspond to the actions of the $n$ polymatrix players on each
side. The $2n$ primary actions are interpreted as $n$ pairs of actions: each
player in the polymatrix game has two actions, and a pair of primary actions
corresponds to these two actions. The top left $2n \times 2n$ primary-primary
block of the bimatrix game encodes the $n^2$ many $2 \times 2$ payoff
matrices in the bipartite polymatrix game. As all the payoffs in the polymatrix
game are in $\{ 0, 1, 2 \}$, all the payoffs in this top-left block are also in
this set.

We call the last $n + 2$ actions of the bimatrix game \textit{secondary}
actions. Outside of the top-left primary-primary block of size $2n \times 2n$,
the remaining game has a \textit{generalized matching pennies} flavour. The
payoffs are constructed in a manner to ensure that (i) the players randomize
uniformly across their primary action-pairs and (ii) each player's choice
between the two actions in a primary action-pair does not depend upon the
other player's strategy over the secondary actions. \cref{tab:bimatrix:overview} shows the payoffs
in the top-right, bottom-left, and bottom-right blocks for the case $n = 3$. We
next describe how this payoff structure achieves the required properties.

\begin{table}[pt]
    \centering
    \renewcommand{\arraystretch}{1.5}
    \begin{tabular}{| c : c | c : c | c : c || c | c | c | c | c | }
        \hline
        \cellcolor{gray!50} & \cellcolor{gray!50} & \cellcolor{gray!50} & \cellcolor{gray!50} & \cellcolor{gray!50} & \cellcolor{gray!50} &
        & \red{$1$} & \blue{$K$} & &
        \\
        \hdashline
        \cellcolor{gray!50} & \cellcolor{gray!50} & \cellcolor{gray!50} & \cellcolor{gray!50} & \cellcolor{gray!50} & \cellcolor{gray!50} &
        & \red{$1$} & \blue{$K$} & &
        \\

        \hline
        \cellcolor{gray!50} & \cellcolor{gray!50} & \cellcolor{gray!50} & \cellcolor{gray!50} & \cellcolor{gray!50} & \cellcolor{gray!50} &
        & & \red{$1$} & \blue{$K$} &
        \\
        \hdashline
        \cellcolor{gray!50} & \cellcolor{gray!50} & \cellcolor{gray!50} & \cellcolor{gray!50} & \cellcolor{gray!50} & \cellcolor{gray!50} &
        & & \red{$1$} & \blue{$K$} &
        \\

        \hline
        \cellcolor{gray!50} & \cellcolor{gray!50} & \cellcolor{gray!50} & \cellcolor{gray!50} & \cellcolor{gray!50} & \cellcolor{gray!50} &
        & & & \red{$1$} & \blue{$K$}
        \\
        \hdashline
        \cellcolor{gray!50} & \cellcolor{gray!50} & \cellcolor{gray!50} & \cellcolor{gray!50} & \cellcolor{gray!50} & \cellcolor{gray!50} &
        & & & \red{$1$} & \blue{$K$}
        \\

        \hline \hline
        \blue{$1$} & \blue{$1$} & & & & & & & & & \red{$1$}
        \\
        \hline
        \red{$K$} & \red{$K$} & \blue{$1$} & \blue{$1$} & & & & & & &
        \\
        \hline
        & & \red{$K$} & \red{$K$} & \blue{$1$} & \blue{$1$} & & & & &
        \\
        \hline
        & & & & \red{$K$} & \red{$K$} & \blue{$1$} & & & &
        \\
        \hline
        & & & & & & \red{$1$} & \blue{$1$} & & &
        \\
        \hline
    \end{tabular}
    \caption{Constructed $\{0, 1, 2, K\}$-valued bimatrix game for $n = 3$, where $K = 8$. The \red{red} (\blue{blue}) values represent the row (column) player's payoffs. The top-left shaded block encodes the polymatrix game payoffs. All white blanks are $0$ payoffs.}
    \label{tab:bimatrix:overview}
\end{table}

We show that all primary action-pairs and all secondary actions get positive probability in any NE. (But an individual action inside a primary action-pair may get zero probability.) At a high level, we do this as follows:
Look at \cref{tab:bimatrix:overview}, and
consider a primary action $i$ of the row player.
For this row $i$, there is exactly one column $j$ where the row player has
payoff $1$  (coloured in red in the table) in the top-right block, all other payoffs in this row are $0$ in the top-right block.

We can show that if the column player gives $0$ probability to column $j$, then
the row player gives $0$ probability to row $i$. In particular, note that the
input polymatrix game has payoff matrices diagonal or anti-diagonal, payoffs
that are at most $2$, and each player's payoff depends upon at most two other players.
Using these properties, we can show that the $i$-th row has at most two
positive payoffs in the top-left block and these positive values are at most
$2$. Then we can show, irrespective of the strategy of the column player
(other than our assumption that the probability of column $j$ is $0$), one of the secondary row actions $i' \ge
2n$ of the row player (specifically, a row with the two $K = 8$ payoffs present) gives strictly
higher utility than the primary row action $i$. Therefore, if column player plays
action $j$ with $0$ probability, then the row player will play action $i$ with
$0$ probability. This idea can be repeated with small modifications to show
that the players give $0$ probability to all rows (and columns), leading to
contradiction. Hence, all primary action-pairs and all secondary actions get
positive probability in every NE.


Observe that for each primary-action pair of the column palyer, there is a
unique row that gives a payoff of $K$ for both actions in this primary
action-pair. As the row player plays all secondary actions with positive
probability in an NE, they must be giving the same expected utility to
the row player, therefore all the primary action-pairs are played uniformly by
the column player. A symmetric argument shows that all primary action-pairs are
played uniformly by the row player as well. On the other hand, notice that the
strategy of a player across their secondary actions does not affect the other
player's choice among the two actions in any given primary action-pair.

Summarizing, we get that the two players randomize uniformly across their primary action-pairs, and a player's choice among the two actions in any primary action-pair does not depend upon the strategy of the other player over her secondary actions.
Hence the players face the same payoff structure as the polymatrix game when making the choice among the two actions in any primary action-pair. This allows us to map the NE strategies of the bimatrix players for their primary actions into an NE of the polymatrix game.

Our reduction above extends a polymatrix to bimatrix reduction by
\citet{ghoshUniqueNash}. Their construction also enforces
uniformity but introduces negative payoffs and payoffs that are $\Theta(n)$. On
the other hand, we improve their construction to get rid of the negative
payoffs and leverage the properties of the input polymatrix game $G$ to ensure
that all the payoffs are in the set $\{ 0, 1, 2, 8 \}$.

Our reduction here also ensures that the resulting bimatrix game is $3$-sparse.
In the initial polymatrix game, the payoff for each player
depended on at most two other players, and that each payoff matrix is diagonal
or anti-diagonal. So each row or column in the top-left block in
the bimatrix game contains at most two non-zero payoffs. Thus, each primary
strategy row or column contains at most three non-zero payoffs, when we take
into account the extra $1$ or $K=8$ payoff that we add in the top-right or bottom-left
blocks (it is important to note here that the $1$ and $K$ payoffs in a primary
row or column lie in different matrices, so they increase the sparsity only by
$1$). In addition to this, each secondary row or column contains either two
$1$s or two $K$s.

\subsection{Column Simulation}

At this point we have a $3$-sparse $\{ 0, 1, 2, 8 \}$-valued bimatrix game,
and our goal is to turn this into a $3$-sparse $\{0, 1\}$-valued bimatrix game.
Here we follow a \emph{column simulation} approach, in which each column of the
game that contains a payoff greater than $1$ is simulated by a set of columns
that have smaller payoffs. This technique was used in prior
work~\cite{abbott2005complexity,chen2007approximation}, with the strongest
prior result being that of \citet{chen2007approximation}
when they showed that finding a $O(1/\poly(n))$-equilibrium in a win-lose game
is PPAD-hard. They present a single general simulation that takes any bimatrix
game and directly produces a win-lose game. However, their simulation does not
produce a game with constant sparsity.
Our approach is to produce three highly-tailored simulations that will take our
$3$-sparse $\{ 0, 1, 2, 8 \}$-valued bimatrix game, and produce a $3$-sparse
win-lose game.

\paragraph{\bf First-type single column simulations.}

The first task is to make payoffs of $8$ in our input game smaller. Suppose that in
our game $(A, B)$, the first column of $A$ is $A_0 = (1, 2, 8, 0 \dots)^\intercal$.
We can simulate this column by building a new game $(A', B')$ in the following
way.

\begin{align*}
A'=\left(
\begin{array}{c c !{\vrule width \arrayrulewidth} c !{\vrule width \arrayrulewidth} c}
2 & 0 & 1 &  \\
0 & 1 & 1 &  \\ \hline
1 & 0 &   & \multirow{5}{*}{$A_{1,\dots,n-1}$} \\
0 & 1 &   &  \\
0 & 4 &   &  \\
0 & 0 &   &  \\
\vdots & \vdots & &
\end{array}
\right)
\qquad
B'=\left(
\begin{array}{c c !{\vrule width \arrayrulewidth} c !{\vrule width \arrayrulewidth} c}
0 & 1 &   &  \\
1 & 0 &   &  \\ \hline
  &   & \multirow{5}{*}{$B_0$}  & \multirow{5}{*}{$B_{1,\dots,n-1}$} \\
  &   &   &  \\
  &   &   &  \\
  &   &   &  \\
\phantom{\vdots}  &   &   &
\end{array}
\right)
\end{align*}

Every blank space in the two matrices above is filled with zeroes.

Here we give a high-level sketch for why this works for exact Nash equilibria.
We will later formally prove that this scheme works for approximate equilibria
with $\eps = O(1/\poly(n))$.
Let us refer to the three column blocks as column block~1, column block~2, and
column block~3,
respectively, and the two row blocks as row block~1 and row block~2.

We will show the following properties hold at an exact Nash equilibrium in the game $(A', B')$.
\begin{enumerate}
\item If the column player allocates positive probability to column block 1, then the
row player must allocate positive probability to row block 1, because otherwise the
column player will receive payoff zero for that column. Moreover, if both
players play their first blocks, then they must mix over both of the strategies
in those blocks.

\item Let $c_1$ and $c_2$ denote the probability assigned to columns $1$ and
$2$ of $A'$. We must have $c_2 = 2 \cdot c_1$. This is to ensure that the row
player is indifferent between the two rows in row block~1 to allow them to mix
over those two rows.

\item This then means that the payoff to the row player arising from columns 1
and 2 of $A'$ are exactly the same as the payoffs that would be obtained when
the column player plays $c_1$ on the first column of $A$. In effect, each
payoff in column 2 is doubled, so the $2$ and the $8$ in $A_0$ are simulated by
the $1$ and the $4$ in column 2 of $A'$.

\item If the column player allocates positive probability to column block~1,
then they must also allocate positive probability to column block~2. This is
because if zero probability was allocated to column block~2, then the third row
of row block 2 would be strictly better than both strategies in row block~1
when the column player plays a strategy with $c_2 = 2 \cdot c_1$. This would in
turn cause the row player to place zero probability on row block~1, and the
contrapositive of point 1 above would then imply that the column player
places zero probability on column block 1, giving a contradiction.

\item The column player cannot allocate all their probability to column
block~2, because if they did, the row player's best-response would be to play
row block~1, and the column player would receive zero payoff.

\item
The row player must allocate positive probability to rows in row block~2. To
see why, first observe that if the column player allocates positive probability
to column block~2, then the row player must place positive probability on row
block~2, because if not, column block~2 would give zero payoff to the column
player. On the other hand, if the column player allocates zero probability to
column block~2, then they also allocate zero probability to column block~1 by
point~4 above. So in this case the column player allocates all their
probability to column block~3, and hence only rows in row block~2 give non-zero
payoff to the row player.
\end{enumerate}

These properties allow us to translate an equilibrium $(x^*, y^*)$ of $(A',
B')$ back to an equilibrium $(x, y)$ of $(A, B)$.
\begin{itemize}
\item To construct $x$ we simply take the probabilities allocated to row
block~2 and re-normalize them to obtain a probability distribution. Point 6
above says that a positive amount of probability is allocated to row block~2,
so this re-normalization is well-defined.

\item To construct $y$, we use $c_1$ as the probability that $y$ assigns
to the first column of $A$, and then we use the probabilities in column block~3
for the rest of the columns. Since we have effectively deleted the probability
assigned to the rest of column block~1 and to column block~2, we then need to
re-normalize to obtain a probability distribution. Point 5 above implies that
this re-normalization is well defined.
\end{itemize}

To see that this is an equilibrium, note that the row player's payoffs in $A$
are the same as the row player's payoffs in row block~2 of $A'$, but scaled by the
re-normalization factor used to define $y$. So any row in row block~2 of $A'$
that is a best-response against $y^*$ continues to be a best-response in $A$
against $y$.
Meanwhile, the payoff to the column player of the columns in
column blocks~2 and~3 in $B'$ are the same as the payoffs in $B$ but scaled by
the re-normalization factor used to define $x$. Point~4 above implies that if
the column player allocates positive probability to column block~1 in $B'$, and
hence positive probability to column 1 in $B$, then they also allocate positive
probability to column block~2. Hence column block~2 in $B'$ is a best-response against
$x^*$, meaning that column 1 in $B$ is a best-response against $x$.

\paragraph{\bf Applying column simulations.}

We can use the method described above to turn the first column of our $\{ 0, 1,
2, 8 \}$-valued bimatrix game into three columns that contain payoffs in the
set $\{ 0, 1, 2, 4 \}$. Note also that the new rows and columns that we
introduce in this way are $3$-sparse. So we can apply this method to each
column of $A$, and then swap the two players, and apply the same method to each
column of $B^\intercal$. The result will be a $\{ 0, 1, 2, 4 \}$-valued bimatrix game
that is $3$-sparse.

Note that the single-column simulation is essentially a generic method that
divides the maximum payoff in a game by two. So the next step is
to take our $\{ 0, 1, 2, 4 \}$-valued bimatrix game, and to apply the
single-column simulation again. This time we turn all of our~$4$ payoffs
into~$2$ payoffs, and we arrive at a $\{ 0, 1, 2 \}$-valued bimatrix game.

While it might seem tempting to apply the column simulation step one more time
to arrive at a win-lose game, this unfortunately does not work. This is because
the method introduces a payoff of $2$ in the upper left corner of the game, and
so we can proceed no further with this technique.

\paragraph{\bf Second-type single column simulations.}

The purpose of a second-type single column simulation is to reduce a $2$ payoff
to two $1$ payoffs. Suppose that we have a game $(A, B)$ in which the first column
of $A$ is $A_0 = (1, 1, 2, 0, \dots)^\intercal$.
We produce the game $(A', B')$ in the following way.

\begin{align*}
A' &=\left(
\begin{array}{c c c  !{\vrule width \arrayrulewidth} c  !{\vrule width \arrayrulewidth} c}
1 & 0 & 0 & 1 & \\
0 & 1 & 0 & 1 & \\
0 & 0 & 1 & 1 & \\ \hline
1 & 0 & 0 &  & \multirow{5}{*}{$A_{1,\dots,n-1}$} \\
0 & 1 & 0 &  &\\
0 & 1 & 1 &  &\\
0 & 0 & 0 &  &\\
\vdots & \vdots & \vdots & &
\end{array}
\right) &
B' &=\left(
\begin{array}{c c c  !{\vrule width \arrayrulewidth} c  !{\vrule width \arrayrulewidth} c}
0 & 1 & 0 &  \\
0 & 0 & 1 &  \\
1 & 0 & 0 &  \\ \hline
  &   &   &\multirow{5}{*}{$B_0$}  & \multirow{5}{*}{$B_{1,\dots,n-1}$} \\
  &   &   &  &  \\
  &   &   &  &  \\
  &   &   &  &  \\
\phantom{\vdots}  &   &   & &
\end{array}
\right)
\end{align*}

This works for essentially the same reason that first-type single column
simulations work. The main difference is that the column player is now required
to mix uniformly over the first three columns. This means that the $2$ payoff
in the original matrix is simulated by the two~$1$ payoffs in the third row of
row block~2. Note also that all the rows and columns introduced
here have sparsity $3$.

\paragraph{\bf Dual column simulations.}

What we have presented so far is unfortunately not enough to get us to a
win-lose bimatrix game that has sparsity~$3$. To see why, recall that our
original game had rows that contained two $8$ payoffs. We then applied
first-type column simulations to first produce a game in which that row
contained two $4$ payoffs, and then to produce a game in which that row
contained two $2$ payoffs. But if we now apply a second-type single column
simulation, each $2$ in the row will be replaced by two $1$s, and we will have
sparsity $4$, not $3$.

We address this by introducing \emph{dual column} simulations, which allow us
to replace a row that contains two $2$s, with a row that contains only one $2$.
Let $(A, B)$ be a game such that
\begin{equation*}
(A_0, A_1) = \begin{pmatrix}
1 & 0 \\
0 & 1 \\
2 & 2 \\
0 & 0 \\
\vdots & \vdots
\end{pmatrix}.
\end{equation*}
We will show that this is the worst case, in the sense that by the time that we need
to apply a dual column simulation, for every $2$ that
shares a row with another $2$, there is at most one $1$ payoff in either of the
columns that contain one of the $2$s.

The dual column simulation of these two columns is as follows.
\begin{align*}
A' &=\left(
\begin{array}{ c c c !{\vrule width \arrayrulewidth} c !{\vrule width \arrayrulewidth} c !{\vrule width \arrayrulewidth} c }
1 & 0 & 0 & 1 &   & \\
0 & 1 & 1 & 1 &   & \\ \hline
1 & 0 & 0 &   & 1 & \\
0 & 1 & 1 &   & 1 & \\ \hline
0 & 1 & 0 &   &   & \multirow{5}{*}{$A_{2, \dots,n-1}$} \\
0 & 0 & 1 &   &   & \\
2 & 0 & 0 &   &   & \\
0 & 0 & 0 &   &   & \\
\vdots & \vdots & \vdots &   &   & \\
\end{array}
\right)&
B' &=\left(
\begin{array}{ c c c !{\vrule width \arrayrulewidth} c !{\vrule width \arrayrulewidth} c !{\vrule width \arrayrulewidth} c }
0 & 1 & 0 &   &   & \\
1 & 0 & 0 &   &   & \\ \hline
0 & 0 & 1 &   &   & \\
1 & 0 & 0 &   &   & \\ \hline
  &   &   & \multirow{5}{*}{$B_0$}  & \multirow{5}{*}{$B_1$}  & \multirow{5}{*}{$B_{2, \dots,n-1}$} \\
  &   &   &   &   & \\
  &   &   &   &   & \\
  &   &   &   &   & \\
\phantom{\vdots} &  &  &   &   & \\
\end{array}
\right)
\end{align*}

Let $c_1$, $c_2$, and $c_3$ be the probabilities that the column player
assigns to the first three columns of $A'$. The key property is that in every
equilibrium we will have $c_1 = c_2 + c_3$, meaning that the column player can
set $c_2$ and $c_3$ independently of one another, but $c_1$ must always be the
sum of the two other probabilities. So the payoffs of the first two rows of row
block~3 depend only on $c_2$ and $c_3$, respectively, but the payoff of the
third row is always $2 \cdot (c_2 + c_3)$. So we are correctly simulating the
original game in which there was a payoff of $2$ in each of those two columns.
In particular, the two $2$s in the original row of $A$ are now replaced by a
single $2$ in the corresponding row of $A'$.

Applying a dual column simulation to each row that contains two $2$s before
applying the second-type single column simulation ultimately gives us
a win-lose bimatrix game with sparsity~$3$.

\paragraph{\bf Polynomial approximations.}

We have described the reduction above in terms of exact equilibria. When formally proving correctness of the reduction, we show that each of the
steps also works for polynomially small WSNEs. The main difference
between what we presented above and the formal reduction is that, when considering approximate equilibria, it is important to simulate all
columns of a matrix in parallel. That is, rather than applying a first-type
single column simulation to each column sequentially, we perform a
single reduction that simulates each column in parallel.

We ultimately show that each simulation step (of any type) increases $\eps$ by a factor of
$O(n^2)$. We apply four first-type simulation steps (two for each
player), two dual column simulation steps, and two second-type
simulation steps, giving us eight simulation steps in total. Thus, if it is
PPAD-hard to find a $O(1/n^2)$-WSNE of the original game, it is PPAD-hard to
find a $O(1/n^{17})$-WSNE of the final win-lose sparsity $3$ game.






\section{Hardness for $3$-Sparse $\{ 0, 1, 2, 8 \}$-Valued Bimatrix Games} \label{sec:res-bimatrix}

In this section, we prove that the problem of computing an $\eps$-WSNE for inverse-polynomially small $\eps > 0$ in $3$-sparse $\{ 0, 1, 2, 8 \}$-valued bimatrix games is PPAD-hard.
Given \cref{lm:ne2wsne}, this implies that computing an $\eps$-NE is also PPAD-hard for sufficiently small inverse-polynomial $\eps$.

We do a reduction from the \PURE{} problem via the $\eps$-\POLY{} problem. The \PURE{} problem is a PPAD-hard fixed-point computation problem in a circuit called the \emph{pure circuit}~\cite{deligkas2024pure}, formally defined in \cref{sec:pure2poly}.
We use $\eps$-\POLY{} to denote the $\eps$-WSNE computation problem in a restricted class of \textit{polymatrix games} defined below.
The restrictions become useful during our subsequent reduction to $3$-sparse $\{ 0, 1, 2, 8 \}$-valued bimatrix games. \citet{deligkas2024pure} also reduce \PURE{} to the equilibrium computation problem of polymatrix games. Our proof has a similar structure, but we modify the constructed polymatrix game to assist our subsequent reduction. Below we formally define the $\eps$-\POLY{} problem and use it to show hardness of computing equilibrium in $3$-sparse $\{ 0, 1, 2, 8 \}$-valued bimatrix games. In \cref{sec:pure2poly}, we formally define the \PURE{} problem and reduce it to the $\eps$-\POLY{} problem.

\begin{definition}[$\eps$-\POLY{}]
In a general polymatrix game, there are $n$ players with $m_i$ actions for player $i \in [n]$.
Let $a_i \in [m_i]$ be the action of player $i$. The payoff of player $i$ from the sub-game associated with player $j \neq i$ is determined by the matrix $A^{ij}$, where player $i$ gets a payoff of $A^{ij}_{a_i a_j}$ from this subgame if $i$ plays $a_i$ and $j$ plays $a_j$. Given the actions of all agents $(a_1, \ldots, a_n)$, the utility of player $i$ is equal to $\sum_{j \neq i} A^{ij}_{a_i a_j}$.
We make the following additional assumptions to restrict the class of polymatrix games we use in our reductions:
\begin{itemize}
    \item Each player has exactly two actions: $a_i \in \{0, 1\}$ for all $i$.

    \item All the payoffs are from the set $\{0, 1, 2\}$: $A^{ij} \in \{0, 1, 2\}^{2 \times 2}$ for all $i$ and $j \neq i$.

    \item All payoff matrices are either diagonal or anti-diagonal: for every $i$ and $j \neq i$, either $A^{ij}_{01} = A^{ij}_{10} = 0$ or $A^{ij}_{00} = A^{ij}_{11} = 0$.


    \item The game is \textit{bipartite}: We can view the polymatrix game as a directed graph, where the vertices are the players, and there is an edge from player $j$ to player $i$ if $A^{ij} \neq 0$. We assume that the underlying undirected graph of this directed graph is bipartite.

    \item The \textit{in-degree} is at most two: The payoff of any player depends upon at most two other players, i.e., for every player $i$, there exist at most two players $k$ and $k'$ such that $A^{ij} = 0$ for all $j \notin \{ k, k' \}$.

    \item The \textit{out-degree} is at most two: Any player can affect the payoff of at most two other players, i.e., for every player $i$, there exist at most two players $k$ and $k'$ such that $A^{ji} = 0$ for all $j \notin \{ k, k' \}$.


\end{itemize}
The $\eps$-\POLY{} problem is to find an $\eps$-WSNE of a polymatrix game satisfying the assumptions given above, for a given $\eps > 0$.
\end{definition}

\begin{theorem}\label{thm:pure2poly}
The $\eps$-\POLY{} problem is PPAD-hard for all $\eps < 1/2$.
\end{theorem}

We prove \cref{thm:pure2poly} by doing a reduction from \PURE{}; the formal definition of \PURE{} and the reduction is given in \cref{sec:pure2poly}.
Notice that the in-degree of each player is at most $2$ and all payoffs lie between $0$ and $2$ in the class of polymatrix games described in the definition of the $\eps$-\POLY{} problem. So, the minimum and maximum utility of each player lies between $0$ and $8$. A normalized (and strategically equivalent) version of this polymatrix game scales down all payoffs by $1/4$, which makes all payoffs in the range $[0, 1]$. \cref{thm:pure2poly} implies PPAD-hardness for $\eps$-WSNE for all $\eps < 1/8$ in the normalized game.

Next we move towards proving PPAD-hardness for computing $\eps$-WSNEs for inverse-poly $\eps$ in a restricted subclass of $3$-sparse $\{ 0, 1, 2, 8 \}$-valued bimatrix games. We introduce these restrictions because they become useful in our subsequent reduction to $3$-sparse win-lose bimatrix games in the next section.

\begin{definition}[$\eps$-\RESBI{m}]
\label{def:resbi}
The bimatrix game $(A, B)$ is $3$-sparse and $\{ 0, 1, 2, 8 \}$-valued, where $A$ and $B$ are $m \times m$ matrices. Further, the rows and columns of the payoff matrices satisfy the following additional properties: Every column of $A$ and $B^\intercal$
\begin{itemize}
    \item either has one or two $1$s,
    \item or has one $8$ and at most two entries from $\{ 1, 2 \}$,
\end{itemize}
and all other entries are $0$s. Every row of $A$ and $B^\intercal$
\begin{itemize}
    \item either has at most three entries from $\{ 1, 2 \}$,
    \item or has exactly two $8$s,
\end{itemize}
and all other entries are $0$s. The $\eps$-\RESBI{m} problem is to find an $\eps$-WSNE of a bimatrix game $(A, B)$ satisfying the assumptions given above, for a given $\eps > 0$.
\end{definition}

\begin{theorem}\label{thm:res-bimatrix}
The $\eps$-\RESBI{m} problem is PPAD-hard for all $\eps < 1/48m$.
\end{theorem}
\begin{proof}
We reduce from the $\delta$-\POLY{} problem for $\delta < 1/2$, which is PPAD-hard as proven in \cref{thm:pure2poly}.
Let $G$ denote the input polymatrix game from the $\delta$-\POLY{} problem.
We will construct a suitable bimatrix game $(A, B)$ based on $G$.
Please note that $G$ satisfies all the constraints described in the definition of $\delta$-\POLY{}, and these constraints of $G$ will be used to get the properties of the bimatrix game $(A, B)$ described in the definition of $\eps$-\RESBI{m}.

Note that the input polymatrix game $G$ is bipartite. For simplicity, let us assume that $G$ has $2n$ players, with $n$ players on each side of the bipartite graph. This is without loss of generality: if there is a side with less players, we add a few dummy players that don't affect the original game.\footnote{
For example, on the smaller side we add players $i$ and $i'$ and on the larger side we add player $j$, with all payoff matrices between these players the $2 \times 2$ identity matrix,
and the payoff matrices with other players
$0$.
By repeating this process we can make the number of players on both sides equal.
Notice that this modification will not affect the behaviour of original players, still satisfies the properties of $\delta$-\POLY{}, and \cref{thm:pure2poly} applies.}
In our constructed bimatrix game $(A, B)$, both player will have $m = 3n + 2$ actions, and the row (column) player will simulate the actions of all the players on the left (right) side of the bipartite graph.

Also note that the input polymatrix game $G$ has both in-degree and out-degree at most $2$.
Let $K = 8 > 4 = 2 \cdot 2$. Our constructed bimatrix game $(A, B)$ will be $3$-sparse and $\{ 0, 1, 2, K \}$-valued.

We use the following notation for the input polymatrix game $G$:
We index the players on each side of the bipartite graph using $[n]$. For player $i \in [n]$ on the left side and player $j \in [n]$ on the right side, and action $s \in \{0,1\}$ of player $i$ and action $t \in \{0,1\}$ of player $j$, let the payoff of player $i$ be denoted by $A^{ij}_{st}$ and of player $j$ be denoted by $B^{ij}_{st}$.  We assume $G$ satisfies all the requirements stated in the definition of $\delta$-\POLY{}.

\paragraph{Construction of the output bimatrix game $(A, B)$.}
The actions of the players are as follows:
Each player has the action set $\{ (i, t) \mid i \in [n],\ t \in \{0, 1\} \} \ \bigcup \ \{ i + n \mid i \in [n + 2] \} $, where the first $2n$ actions of type $(i, t)$ are called the \textit{primary} actions and the last $n + 2$ actions are called the \textit{secondary} actions.
We call $i$ in $(i, t)$ a \textit{primary index}, which will correspond to the $i$-th player on one of the sides of $G$ (of the left or the right side of the underlying bipartite graph of $G$ depending upon whether we are looking at the row or the column player of $(A, B)$, respectively).
We call $t$ in $(i, t)$ the \textit{bit}, which will correspond to player $i$'s action in $G$.
Finally, we call an action $i$ among the last $n + 2$ actions a \textit{secondary index}.


We embed the polymatrix game $G$ in the top-left block of $(A, B)$ corresponding to the $2n$ primary actions. We set the payoffs of the secondary actions to ensure that the two players mix approximately uniformly over their primary indices. This almost uniformity over the primary indices ensures that when choosing the bit corresponding to each primary index, the players face approximately the same payoff structure as the polymatrix game (scaled by a positive constant). This construction to enforce (approximate) uniformity is based on the construction of \citet{ghoshUniqueNash}. Their construction also enforces uniformity but introduces negative payoffs and payoffs that are $\Theta(n)$.
We improve their construction to get rid of the negative payoffs and leverage the properties of the input polymatrix game $G$ to ensure
that all the payoffs are in the set $\{ 0, 1, 2, 8 \}$.

Let $N = 2n + 2$, the number of indices. Recall that $[k] = \{ 0, 1, \ldots, k-1 \}$, and $\bmod(\cdot, k) \in [k]$ for any positive integer $k$.
We next construct the payoff matrices $A$ and $B$ of the bimatrix game, where $A_{\alpha, \beta}$ and $B_{\alpha, \beta}$ denote the payoffs of the row and column players, when the row and column players play actions $\alpha$ and $\beta$, respectively.
\begin{itemize}
    \item Both the row and column players play primary actions $(i,s)$ and $(j,t)$, respectively. Note that $i, j \in [n]$ and $s, t \in [2]$.
    \[
        A_{(i,s),(j,t)} = A^{ij}_{st}, \qquad B_{(i,s),(j,t)} = B^{ij}_{st}.
    \]



    \item The row player plays a primary action $(i, s)$ and the column player plays a secondary action $j$. Note that $(i, s) \in [n] \times [2]$ and $j \in [N] \setminus [n]$.
    \begin{align*}
        A_{(i,s),j} &= \begin{cases}
            1, &\text{if $\bmod(j - i, N) = n + 1$},\\
            0, &\text{otherwise,}
        \end{cases} \quad
        B_{(i,s),j} = \begin{cases}
            K, &\text{if $\bmod(j - i, N) = n + 2$},\\
            0, &\text{otherwise.}
        \end{cases}
    \end{align*}

    \item The row player plays a secondary action $i$ and the column player plays a primary action $(j, t)$. Note that $i \in [N] \setminus [n]$ and $(j, t) \in [n] \times [2]$.
    \begin{align*}
        A_{i,(j,t)} &= \begin{cases}
            K, &\text{if $\bmod(j - i, N) = n + 1$},\\
            0, &\text{otherwise,}
        \end{cases} \quad
        B_{i,(j,t)} = \begin{cases}
            1, &\text{if $\bmod(j - i, N) = n + 2$},\\
            0, &\text{otherwise.}
        \end{cases}
    \end{align*}

    \item Both the row and column players play secondary actions $i$ and $j$, respectively. Note that $i, j \in [N] \setminus [n]$.
    \begin{align*}
        A_{i,j} &= \begin{cases}
            1, &\text{if $\bmod(j - i, N) = n + 1$},\\
            0, &\text{otherwise,}
        \end{cases} \quad
        B_{i,j} = \begin{cases}
            1, &\text{if $\bmod(j - i, N) = n + 2$},\\
            0, &\text{otherwise.}
        \end{cases}
    \end{align*}
\end{itemize}
In other words, when both players play primary actions, then the payoffs encode the polymatrix game $G$'s payoffs. But when either player plays a secondary action, then the row player gets a positive payoff in $\{ 1, K \}$ if and only if $\bmod(j - i, N) = n + 1$ and the column player gets a positive payoff in $\{ 1, K \}$ if and only if $\bmod(j - i, N) = n + 2$, where $i$ and $j$ are the indices played by the two players, otherwise the payoffs are zero. \cref{tab:bimatrix} shows the payoff matrices of the constructed bimatrix game $(A, B)$ for $n = 3$.
It is not difficult to check the constructed bimatrix game $(A, B)$ satisfies the constraints of $\eps$-\RESBI{m}; formally shown in \cref{app:resbi-check}.

\begin{table}[]
    \centering
    \renewcommand{\arraystretch}{1.5}
    \begin{tabular}{| c : c | c : c | c : c | c | c | c | c | c | }
        \hline
        \red{$A^{00}_{00}$} \blue{$B^{00}_{00}$} & \red{$A^{00}_{01}$} \blue{$B^{00}_{01}$} &
        \red{$A^{01}_{00}$} \blue{$B^{01}_{00}$} & \red{$A^{01}_{01}$} \blue{$B^{01}_{01}$} &
        \red{$A^{02}_{00}$} \blue{$B^{02}_{00}$} & \red{$A^{02}_{01}$} \blue{$B^{02}_{01}$} &
        & \red{$1$} & \blue{$K$} & &
        \\
        \hdashline
        \red{$A^{00}_{10}$} \blue{$B^{00}_{10}$} & \red{$A^{00}_{11}$} \blue{$B^{00}_{11}$} &
        \red{$A^{01}_{10}$} \blue{$B^{01}_{10}$} & \red{$A^{01}_{11}$} \blue{$B^{01}_{11}$} &
        \red{$A^{02}_{10}$} \blue{$B^{02}_{10}$} & \red{$A^{02}_{11}$} \blue{$B^{02}_{11}$} &
        & \red{$1$} & \blue{$K$} & &
        \\

        \hline
        \red{$A^{10}_{00}$} \blue{$B^{10}_{00}$} & \red{$A^{10}_{01}$} \blue{$B^{10}_{01}$} &
        \red{$A^{11}_{00}$} \blue{$B^{11}_{00}$} & \red{$A^{11}_{01}$} \blue{$B^{11}_{01}$} &
        \red{$A^{12}_{00}$} \blue{$B^{12}_{00}$} & \red{$A^{12}_{01}$} \blue{$B^{12}_{01}$} &
        & & \red{$1$} & \blue{$K$} &
        \\
        \hdashline
        \red{$A^{10}_{10}$} \blue{$B^{10}_{10}$} & \red{$A^{10}_{11}$} \blue{$B^{10}_{11}$} &
        \red{$A^{11}_{10}$} \blue{$B^{11}_{10}$} & \red{$A^{11}_{11}$} \blue{$B^{11}_{11}$} &
        \red{$A^{12}_{10}$} \blue{$B^{12}_{10}$} & \red{$A^{12}_{11}$} \blue{$B^{12}_{11}$} &
        & & \red{$1$} & \blue{$K$} &
        \\

        \hline
        \red{$A^{20}_{00}$} \blue{$B^{20}_{00}$} & \red{$A^{20}_{01}$} \blue{$B^{20}_{01}$} &
        \red{$A^{21}_{00}$} \blue{$B^{21}_{00}$} & \red{$A^{21}_{01}$} \blue{$B^{21}_{01}$} &
        \red{$A^{22}_{00}$} \blue{$B^{22}_{00}$} & \red{$A^{22}_{01}$} \blue{$B^{22}_{01}$} &
        & & & \red{$1$} & \blue{$K$}
        \\
        \hdashline
        \red{$A^{20}_{10}$} \blue{$B^{20}_{10}$} & \red{$A^{20}_{11}$} \blue{$B^{20}_{11}$} &
        \red{$A^{21}_{10}$} \blue{$B^{21}_{10}$} & \red{$A^{21}_{11}$} \blue{$B^{21}_{11}$} &
        \red{$A^{22}_{10}$} \blue{$B^{22}_{10}$} & \red{$A^{22}_{11}$} \blue{$B^{22}_{11}$} &
        & & & \red{$1$} & \blue{$K$}
        \\

        \hline
        \blue{$1$} & \blue{$1$} & & & & & & & & & \red{$1$}
        \\
        \hline
        \red{$K$} & \red{$K$} & \blue{$1$} & \blue{$1$} & & & & & & &
        \\
        \hline
        & & \red{$K$} & \red{$K$} & \blue{$1$} & \blue{$1$} & & & & &
        \\
        \hline
        & & & & \red{$K$} & \red{$K$} & \blue{$1$} & & & &
        \\
        \hline
        & & & & & & \red{$1$} & \blue{$1$} & & &
        \\
        \hline
    \end{tabular}
    \caption{Constructed $\{0, 1, 2, K\}$-valued bimatrix game for $n = 3$. The \red{red} (\blue{blue}) values represent the row (column) player's payoffs. The blanks are zero payoffs.}
    \label{tab:bimatrix}
\end{table}

\paragraph{Analysis.}
Pick any $\delta < 1/2$. Let $\eps = \delta/3KN < 1/6KN$.
Let $(x, y)$ denote an $\eps$-WSNE of the bimatrix game $(A, B)$, where $x$ and $y$ are indexed as follows:
\begin{itemize}
    \item Let $x_{is}$ be the probability that the row player plays primary action $(i, s)$. Let $x_i = x_{i0} + x_{i1}$ and $p_i = x_{i1} / x_i$ (arbitrary if $x_i = 0$).
    Similarly, let $y_{jt}$ be the probability that the column player plays primary action $(j, t)$, $y_j = y_{j0} + y_{j1}$, and $q_j = y_{j1} / y_j$.
    We shall prove that the $x_i$'s and $y_j$'s are all $\Theta(1/n)$ and almost equal across $i \in [n]$ and $j \in [n]$, and $(p, q)$ is a $\delta$-WSNE of the polymatrix game $G$.

    \item Let $x_i$ and $y_j$ be the probabilities that the row and column players, respectively, play secondary actions $i, j \in [N] \setminus [n]$. Note that $N = 2n + 2$, the number of indices.
\end{itemize}

We next present three lemmas that help us prove our theorem.
The first lemma shows that the two players play each of their indices with sufficiently large probability. The second one proves that the two players randomize almost uniformly over their primary indices. And the third lemma shows that $(p,q)$ is an approximate equilibrium of the polymatrix game. The proofs of the lemmas are provided in \cref{app:positive-prob-proof}, \cref{app:unif-over-indices-proof}, and \cref{app:encode-proof}, respectively; the technical overview in \cref{sec:to} has a high-level sketch of the arguments.

\begin{lemma}\label{lm:positive-prob}
$x_i, y_i \ge 1/2KN$ for all $i \in [N]$.
\end{lemma}


\begin{lemma}\label{lm:unif-over-indices}
$|x_i - x_j| \le \eps/K$ and $|y_i - y_j| \le \eps/K$ for all $i,j \in [n]$.
\end{lemma}


\begin{lemma}\label{lm:encode}
$(p,q)$ is a $\delta$-WSNE of the polymatrix game $G$.
\end{lemma}

Note that we assumed $\delta < 1/2$ and $\eps = \delta / 3 K N$.
\cref{lm:encode} says that if $(x, y)$ is an $\eps$-WSNE of the bimatrix game $(A, B)$ then $(p, q)$ is a $\delta$-WSNE of the polymatrix game $G$.
Computing $\delta$-WSNE of $G$ is PPAD-hard for any $\delta < 1/2$.
So, computing $\eps$-WSNE of $(A, B)$ is PPAD-hard for any $\eps = \delta/3KN < 1/6KN = 1/48N$.
As $m = 3n + 2 > 2n + 2 = N$, we get our required result.
\end{proof}

\section{Reduction from $\{0, 1, 2, 8\}$-Valued Bimatrix to $\{ 0, 1\}$-Valued Bimatrix}

\subsection{Type One Single Column Simulations}


As mentioned in the technical overview, we will apply type one single column
simulations twice. The first application will be to a game that satisfies
Definition~\ref{def:resbi}.
The second application will be to a game that satisfies the following
properties.
\begin{definition}[Stage 1 game]
\label{def:stage1}
A game $(A, B)$ is a stage 1 game if all of the following hold.
Every column of $A$ and $B^\intercal$ contains either
(i) At most three non-zero payoffs in the set $\{1, 2\}$ with at most one
$2$, or
(ii) A single $4$ payoff along with at most three $1$s.
Every row of $A$ and $B^\intercal$ contains either
(i) At most three $1$ payoffs, or
(ii) A single $2$ along with at most one $1$, or
(iii) Exactly two $4$ payoffs.
\end{definition}

Our presentation will be generic over these two types of input game. We will
use $K$ to denote either $8$ or $4$.
We assume, without loss of generality, that the columns of $A$ that do not
contain a value $K$ come after the other columns of $A$, and we use $k$ to denote
the index of the first column of $A$ that does not contain a value $K$.

Let $(S, T)$ be a bimatrix game defined in the following way.
\begin{align*}
S&= \begin{pmatrix}
2 & 0 \\
0 & 1
\end{pmatrix}
&
T&= \begin{pmatrix}
0 & 1 \\
1 & 0
\end{pmatrix}
\end{align*}
We define the function $\encode$ that maps columns of length $n$ from $A$ to
matrices of size $2 \times n$ in the following way. In what follows, without loss of
generality, we will show the input vectors with their non-zero values
appearing first, and with those values sorted in ascending order
\begin{align*}
\encode\begin{pmatrix}
K \\
0 \\
\vdots
\end{pmatrix} &=
\begin{pmatrix}
0 & K/2 \\
0 & 0 \\
\vdots & \vdots
\end{pmatrix} &
\encode\begin{pmatrix}
1 \\
K \\
0 \\
\vdots
\end{pmatrix} &=
\begin{pmatrix}
1 & 0 \\
0 & K/2 \\
0 & 0 \\
\vdots & \vdots
\end{pmatrix}
&
\encode\begin{pmatrix}
1 \\
1 \\
K \\
0 \\
\vdots
\end{pmatrix} &=
\begin{pmatrix}
1 & 0 \\
1 & 0 \\
0 & K/2 \\
0 & 0 \\
\vdots & \vdots
\end{pmatrix} &
\encode\begin{pmatrix}
2 \\
K \\
0 \\
\vdots
\end{pmatrix} &=
\begin{pmatrix}
0 & 1 \\
0 & K/2 \\
0 & 0 \\
\vdots & \vdots
\end{pmatrix} \\
\encode\begin{pmatrix}
1 \\
2 \\
K \\
0 \\
\vdots
\end{pmatrix} &=
\begin{pmatrix}
1 & 0 \\
0 & 1 \\
0 & K/2 \\
0 & 0 \\
\vdots & \vdots
\end{pmatrix}
&
\encode\begin{pmatrix}
2 \\
2 \\
K \\
0 \\
\vdots
\end{pmatrix} &=
\begin{pmatrix}
0 & 1 \\
0 & 1 \\
0 & K/2 \\
0 & 0 \\
\vdots & \vdots
\end{pmatrix} &
\encode\begin{pmatrix}
1 \\
1 \\
1 \\
K \\
0 \\
\vdots
\end{pmatrix} &=
\begin{pmatrix}
1 & 0 \\
1 & 0 \\
1 & 0 \\
0 & K/2 \\
0 & 0 \\
\vdots & \vdots
\end{pmatrix}
\end{align*}
Observe that this definition captures all possibilities allowed by
Definitions~\ref{def:resbi} and~\ref{def:stage1}.

We define the bimatrix game $(A', B')$ in the following way.

\begin{equation*}
A' = \begin{pmatrix}
S & 1 & 0 & 0 & \dots & 0 & 0 & 0 \\
0 & 0 & S & 1 & \dots & 0 & 0 & 0 \\
\vdots & \vdots & \vdots & \vdots & \ddots & \vdots & \vdots & \vdots \\
0 & 0 & 0 & 0 & \dots & S & 1 & 0 \\
\encode(A_0) & 0 & \encode(A_1) & 0 & \dots & \encode(A_{k-1}) & 0 & A_{k, \dots,
n-1}
\end{pmatrix}
\end{equation*}

\begin{equation*}
B' = \begin{pmatrix}
T & 0 & 0 & 0 & \dots & 0 & 0 & 0 \\
0 & 0 & T & 0 & \dots & 0 & 0 & 0 \\
\vdots & \vdots & \vdots & \vdots & \ddots & \vdots & \vdots & \vdots \\
0 & 0 & 0 & 0 & \dots & T & 0 & 0 \\
0 & B_0 & 0 & B_1 & \dots & 0 & B_{k-1} & B_{k, \dots, n-1}
\end{pmatrix}
\end{equation*}

We divide the rows and columns of $A'$ and
$B'$ into blocks.
\begin{itemize}
\item For each column $j < k$ of $A$ we define $\cb_1^j = \{3j, 3j + 1\}$
and $\cb_2^j = \{ 3j + 2 \}$ as the first and second column blocks for
column $j$.

\item For each column $j < k$ of $A$ we define $\rb^j = \{2j, 2j + 1\}$ to be the row block for column $j$.

\item We also define $\rb_e = \{2k, 2k + 1, \dots 2k + n -1\}$ to be the
\emph{encoding} row block.
\end{itemize}

We present a formal definition of the matrices used in this reduction in Appendix~\ref{app:matrices}.

\paragraph{\bf Translating back.}

Let $(x^*, y^*)$ be an $\eps$-WSNE of $(A', B')$. We construct the strategy
profile $(x, y)$ for $(A, B)$ in the following way.
Given a vector $v \in \reals^n$, we say that $v' \in \reals^n$ is the
\emph{re-normalization} of $v$ if $v'_i = v_i / \sum_{j = 1}^n v_j$ for all $i$,
meaning that $v'$ sums to $1$. Clearly the re-normalization operation requires
that $\sum_{j = 1}^n v_j \ne 0$.

To construct the strategy $x$ we first construct $x'$ so that $x'_i = x^*_{2k +
i}$ for all $i$, and then re-normalize $x'$ to produce $x$. In other words,
$x$ takes the probabilities that $x^*$ assigns to $\rb_e$, and then
renormalizes them into a strategy. We show in Lemma~\ref{lem:rbelower}
that the re-normalization used here is well defined.

To construct the strategy $y$, we use the following procedure.
First construct $y'$ so that
\begin{equation*}
y'_j = \begin{cases}
y^*_{3j} & \text{if $j < k$ and $P(\cb_1^j) \ge 11 \eps$,} \\
0 & \text{if $j < k$ and $P(\cb_1^j) < 11 \eps$,} \\
y^*_{3k + j} & \text{otherwise.}
\end{cases}
\end{equation*}
Then re-normalize $y'$ to produce $y$.
Lemma~\ref{lem:clowerfinal} shows that the re-normalization used here
is well defined.
The proof of the following theorem is shown in Appendix~\ref{app:scs1proof}.
\begin{theorem}
\label{thm:scs1}
If $4 \le K \le 8$, if $\eps < 1/(240 \cdot n^2)$, and if $(x^*, y^*)$ is an
$\eps$-WSNE of $(A', B')$, then both $x$ and $y$ are well-defined, and $(x, y)$
is an $O(n^2 \cdot \eps)$-WSNE of $(A, B)$.
\end{theorem}

\paragraph{\bf The payoff structure.}

We now consider how the reduction affects the payoffs that appear in the game
$(A', B')$. The following lemma considers the matrix $A'$.
The following two lemmas are shown in Appendix~\ref{app:payoffs}.
\begin{lemma}
\label{lem:scs1ap}
If $(A, B)$ satisfies Definition~\ref{def:resbi}, then $A'$ satisfies the
requirements of a stage 1 game.
\end{lemma}


\begin{lemma}
\label{lem:scs1bp}
The matrix $(B')^\intercal$ differs from $B^\intercal$ only by adding rows and columns
that contain one~$1$.
\end{lemma}

We now consider the outcome of applying a first type single column simulation
to a stage 1 game. We will prove that the result satisfies the following
definition.

\begin{definition}[Stage 2 game]
\label{def:stage2}
A game $(A, B)$ is a stage 2 game if it satisfies the following properties.
Every column of $A$ and $B^\intercal$ contains
at most four non-zero payoffs in the set $\{1, 2\}$ with at most one $2$.
Every row of $A$ and $B^\intercal$ contains either
(i) At most three $1$ payoffs, or
(ii) A single $2$ along with at most one $1$, or
(iii) Exactly two $2$ payoffs.
Furthermore, if a row contains exactly two $2$ payoffs, then those $2$s share a
column with exactly one $1$.
\end{definition}

The following lemma is shown in Appendix~\ref{app:payoffs}.
There is no need to prove a separate lemma for the matrix $B'$, because Lemma~\ref{lem:scs1bp} already covers this case.

\begin{lemma}
\label{lem:scs1app}
If $(A, B)$ is a stage 1 game, then $A'$ satisfies the requirements of a stage
2 game.
\end{lemma}

\subsection{Dual Column Simulations}

The dual column simulation takes as input a stage 2 game. Let $I$ be the set of
rows of $A$ that contain exactly two $2$s. We assume, without loss of
generality, that each row in $I$ appears before each row not in $I$. We use $k$
to denote the first row of $A$ that is not in $I$. We also assume that each row
$i \in I$ has its two $2$s in columns $2i$ and $2i + 1$. These properties can
easily be ensured by reordering the rows and columns of $A$.

We will use the following matrices in the reduction.

\begin{align*}
S &= \begin{pmatrix}
1 & 0 & 0 \\
0 & 1 & 1 \\
\end{pmatrix}
&T_1 &= \begin{pmatrix}
0 & 1 & 0 \\
1 & 0 & 0 \\
\end{pmatrix}
&T_2 &= \begin{pmatrix}
0 & 0 & 1 \\
1 & 0 & 0
\end{pmatrix}
\end{align*}

We define the function $\encode$ that maps pairs of columns of length $n$ from $A$ to
matrices of size $3 \times n$ as follows. By our assumptions,
there is only one type of column pair to consider.
\begin{equation*}
\encode \begin{pmatrix}
1 & 0 \\
0 & 1 \\
2 & 2 \\
0 & 0 \\
\vdots & \vdots
\end{pmatrix}
=
\begin{pmatrix}
0 & 1 & 0 \\
0 & 0 & 1 \\
2 & 0 & 0 \\
0 & 0 & 0 \\
\vdots & \vdots & \vdots
\end{pmatrix}
\end{equation*}
Finally we define the output of the reduction to be the game $(A', B')$ as
follows.

\setcounter{MaxMatrixCols}{15}
\begin{equation*}
A' = \begin{pmatrix}
S & 1 & 0 & 0 & 0 & 0 & \dots & 0 & 0 & 0 & 0 \\
S & 0 & 1 & 0 & 0 & 0 & \dots & 0 & 0 & 0 & 0 \\
0 & 0 & 0 & S & 1 & 0 & \dots & 0 & 0 & 0 & 0 \\
0 & 0 & 0 & S & 0 & 1 & \dots & 0 & 0 & 0 & 0 \\
\vdots & \vdots & \vdots & \vdots & \vdots & \vdots & \ddots & \vdots & \vdots
& \vdots & \vdots \\
0 & 0 & 0 & 0 & 0 & 0 & \dots & S & 1 & 0 & 0 \\
0 & 0 & 0 & 0 & 0 & 0 & \dots & S & 0 & 1 & 0 \\
\encode(A_{0, 1}) & 0 & 0 & \encode(A_{2, 3}) & 0 & 0 & \dots &
\encode(A_{2(k-1), 2(k-1) + 1}) & 0 & 0 & A_{2k, \dots, n-1}
\end{pmatrix}
\end{equation*}

\begin{equation*}
B' = \begin{pmatrix}
T_1 & 0 & 0 & 0 & 0 & 0 & \dots & 0 & 0 & 0 & 0 \\
T_2 & 0 & 0 & 0 & 0 & 0 & \dots & 0 & 0 & 0 & 0 \\
0 & 0 & 0 & T_1 & 0 & 0 & \dots & 0 & 0 & 0 & 0 \\
0 & 0 & 0 & T_2 & 0 & 0 & \dots & 0 & 0 & 0 & 0 \\
\vdots & \vdots & \vdots & \vdots & \vdots & \vdots & \ddots & \vdots & \vdots
& \vdots & \vdots \\
0 & 0 & 0 & 0 & 0 & 0 & \dots & T_1 & 0 & 0 & 0 \\
0 & 0 & 0 & 0 & 0 & 0 & \dots & T_2 & 0 & 0 & 0 \\
0 & B_0 & B_1 & 0 & B_2 & B_3 & \dots & 0 & B_{2(k-1)} & B_{2(k-1) + 1} &
B_{2k, \dots, n-1}
\end{pmatrix}
\end{equation*}

To define these matrices formally, we first define the following column and row
blocks.
\begin{itemize}
\item For each row $i < k$, we define row block $\rb_1^i$ to contain rows $4i$
and $4i + 1$, row block $\rb_2^i$ to contain rows $4i + 2$ and $4i + 3$.

\item We define the \emph{encoding} row block $\rb_e$ to contain rows $4k$
through $4k + n-1$.

\item For each row $i < k$, we define column block $\cb_1^i$ to contain columns
$5i$ through $5i + 2$, column block $\cb_2^i$ to contain the column $5i + 3$,
and column block $\cb_3^i$ to contain the column $5i + 4$.
\end{itemize}
We give a formal definition of the matrices above in
Appendix~\ref{app:matrices}.

\paragraph{\bf Translating back.}

Let $(x^*, y^*)$ be an $\eps$-WSNE of $(A', B')$. We construct a strategy
profile $(x, y)$ for $(A, B)$ in the following way.
To construct $x$, we first construct $x'$ so that $x'_i = x^*_{4k + (i - k)}$,
and then we re-normalize $x'$ to produce $x$. In other words, the strategy $x$
plays a re-normalized version of the strategy that $x^*$ plays over the
encoding row block $\rb_e$.  The validity of the re-normalization that we use
here will later be
shown in Lemma~\ref{lem:dcsxlower}.

To construct $y$, we first construct $y'$ as follows.
\begin{equation*}
y'_j = \begin{cases}
y^*_{5\lfloor j/2 \rfloor +1} & \text{if $\lfloor j/2 \rfloor < k$ and $j$ is
even,} \\
y^*_{5 \lfloor j/2 \rfloor +2} & \text{if $\lfloor j/2 \rfloor < k$ and $j$ is
odd,} \\
y^*_{5k + (j - 2k)} & \text{otherwise.}
\end{cases}
\end{equation*}
Next we construct $y''$ from $y'$ as follows.
\begin{equation*}
y''_j = \begin{cases}
0 & \text{if $\lfloor j/2 \rfloor < k$ and $y'_j < 3 \eps$,} \\
y'_j & \text{otherwise.}
\end{cases}
\end{equation*}
Finally we re-normalize $y''$ to produce $y$. The validity of the
re-normalization that we use here will later be shown in
Lemma~\ref{lem:dcsyplower}.
The following theorem will be shown in Appendix~\ref{app:dcsproof}.
\begin{theorem}
\label{thm:dcs}
If $\eps < 1/(600 n^2)$ and $(x^*, y^*)$ is an $\eps$-WSNE of $(A', B')$, then
both $x$ and $y$ are well defined, and
$(x, y)$ is an $O(n^2 \cdot \eps)$-WSNE of $(A, B)$.
\end{theorem}

\paragraph{\bf The payoff structure.}

We now consider how the reduction affects the payoffs that appear in the game
$(A', B')$. The following definition captures the desired output of the
reduction.

\begin{definition}[Stage 3 game]
\label{def:stage3}
A game $(A, B)$ is a stage 3 game if
every column of $A$ and $B^\intercal$ contains
at most four non-zero payoffs in the set $\{1, 2\}$ with at most one $2$, and
if
every row of $A$ and $B^\intercal$ contains either
(i) At most three $1$ payoffs, or
(ii) A single $2$ along with at most one $1$.
\end{definition}

The following pair of lemmas is shown in Appendix~\ref{app:payoffs}.

\begin{lemma}
\label{lem:dcsapayoffs}
If $(A, B)$ is a stage 2 game, then $A'$ satisfies the requirements of a stage
3 game.
\end{lemma}

\begin{lemma}
\label{lem:dcsbpayoffs}
$(B')^\intercal$ differs from $B$ only by adding rows and columns that
have at most two~$1$s.
\end{lemma}

\subsection{Type Two Single Column Simulations}

A type two single column simulation takes as input a stage 3 game $(A, B)$.
We assume without loss of
generality that the columns of $A$ are ordered so that each column that does
not contain at most three $1$s (and no $2$s)
comes before each column that does contain at most three $1$s (and no $2$s), and we use
$k$ to denote the first column of $A$ that does contain at most three $1$s.

The reduction is almost identical to a type one single column simulation, but with
the $S$ and $T$ matrices changed, and the $\encode$ function redefined.  We define $S$ and $T$ in the following way.
\begin{align*}
S &= \begin{pmatrix}
1 & 0 & 0 \\
0 & 1 & 0 \\
0 & 0 & 1
\end{pmatrix}
&
T &= \begin{pmatrix}
0 & 1 & 0 \\
0 & 0 & 1 \\
1 & 0 & 0
\end{pmatrix}
\end{align*}
The $\encode$ function is defined as follows, assuming without loss of
generality that the input to the function has non-zero payoffs
before the zero payoffs, and the former are listed in increasing order.

\begin{align*}
\encode \begin{pmatrix}
2 \\
0 \\
\vdots
\end{pmatrix} &=
\begin{pmatrix}
0 & 1 & 1 \\
0 & 0 & 0 \\
\vdots & \vdots & \vdots \\
\end{pmatrix}
&
\encode \begin{pmatrix}
1 \\
2 \\
0 \\
\vdots
\end{pmatrix} &=
\begin{pmatrix}
1 & 0 & 0 \\
0 & 1 & 1 \\
0 & 0 & 0 \\
\vdots & \vdots & \vdots \\
\end{pmatrix}
&
\encode \begin{pmatrix}
1 \\
1 \\
2 \\
0 \\
\vdots
\end{pmatrix} &=
\begin{pmatrix}
1 & 0 & 0 \\
1 & 0 & 0 \\
0 & 1 & 1 \\
0 & 0 & 0 \\
\vdots & \vdots & \vdots \\
\end{pmatrix} \\
\encode \begin{pmatrix}
1 \\
1 \\
1 \\
2 \\
0 \\
\vdots
\end{pmatrix} &=
\begin{pmatrix}
1 & 0 & 0 \\
1 & 0 & 0 \\
0 & 1 & 0 \\
0 & 1 & 1 \\
0 & 0 & 0 \\
\vdots & \vdots & \vdots \\
\end{pmatrix}
&
\encode \begin{pmatrix}
1 \\
1 \\
1 \\
1 \\
0 \\
\vdots
\end{pmatrix} &=
\begin{pmatrix}
1 & 0 & 0 \\
1 & 0 & 0 \\
0 & 1 & 0 \\
0 & 0 & 1 \\
0 & 0 & 0 \\
\vdots & \vdots & \vdots \\
\end{pmatrix}
\end{align*}
Other than these changes, and the corresponding changes needed to $\cb_1^j$ and
$\rb^j$ to include three columns or rows instead of two, the definition of the
reduction is identical to the definition of a type one single column
simulation.

\paragraph{\bf Translating back.}

Given an $\eps$-WSNE $(x^*, y^*)$ of $(A', B')$ we define the strategy profile
$(x, y)$ for $(A, B)$ in the same way as we did for first type simulations.
The only difference is that the threshold at which a column's probability is
zeroed is $18 \eps$ for this reduction, whereas it was $11 \eps$ in the previous
one.
The proof of the following theorem is similar to that of
Theorem~\ref{thm:scs1}, and is given in Appendix~\ref{app:scs2}.
\begin{theorem}
\label{thm:scs2}
If $\eps < 1/(96 \cdot n^2)$, and if
$(x^*, y^*)$ is an $\eps$-WSNE of $(A', B')$, then both $x$ and $y$ are
well-defined, and $(x, y)$ is an $O(n^2 \cdot \eps)$-WSNE of $(A, B)$.
\end{theorem}

\paragraph{\bf The payoff structure.}

The following pair of lemmas is shown in
Appendix~\ref{app:payoffs}.

\begin{lemma}
\label{lem:scs2bp}
Each row and column of $A'$ contains at most three $1$s.
\end{lemma}


\begin{lemma}
\label{lem:scs2ap}
$(B')^\intercal$ differs from $B$ only in the addition of rows
and columns that contain one $1$.
\end{lemma}

\subsection{Reducing to a $\{ 0, 1 \}$-Valued Bimatrix Game}

We can now apply our column simulation reductions to turn the bimatrix game
from Section~\ref{sec:res-bimatrix} to a $\{ 0, 1 \}$-valued bimatrix game, while carefully ensuring that
the resulting game is 3-sparse.
We start with the game $(A, B)$, which is an instance of $\eps$-\RESBI{m}.

We apply a type-one single column simulation to $(A_1, B_1)$. We then
swap the players and apply a type-one single column simulation to the column
player, and then swap the players back again.
Lemmas~\ref{lem:scs1ap} and~\ref{lem:scs1bp} imply that the resulting game
is a stage 1 game.
We then apply another type-one single column simulation in the same way, and
Lemmas~\ref{lem:scs1app} and~\ref{lem:scs1bp} imply that the resulting game is a
stage 2 game. We apply a dual column simulation in the same way, and
Lemmas~\ref{lem:dcsapayoffs} and~\ref{lem:dcsbpayoffs} imply that
the resulting game is a stage 3 game. Finally, we apply a type-two single
column simulation in the same way, and
Lemmas~\ref{lem:scs2bp} and~\ref{lem:scs2ap} ensure that the resulting game is
a $3$-sparse win-lose game.





Theorems~\ref{thm:scs1},~\ref{thm:scs2}, and~\ref{thm:dcs} imply that each
application of a column simulation reduction blows up $\eps$ by a factor of
$O(n^2)$. Our four steps above each use two simulations, giving eight
simulations in total, so the total blow-up in $\eps$ is $O(n^{16})$. Thus, if
one can find an $\eps$-WSNE of $(A_5, B_5)$, then one can
recover an $O(n^{16} \cdot \eps)$-WSNE of $(A, B)$.
Since we have shown that it is PPAD-hard to find a $O(1/n)$-WSNE of $(A, B)$,
it follows that it is PPAD-hard to find a $O(1/n^{17})$-WSNE of $(A', B')$.
So we have shown the following theorem.

\begin{theorem}
It is PPAD-hard to find an $\eps$-WSNE for $\eps \in O(1/\poly(n))$ for a
zero-one bimatrix game that is 3-sparse.
\end{theorem}

Applying Lemma~\ref{lm:ne2wsne} then immediately gives us a lower bound for
$\eps$-NE as well.

\begin{corollary}
It is PPAD-hard to find an $\eps$-NE for $\eps \in O(1/\poly(n))$ for a
zero-one bimatrix game that is 3-sparse.
\end{corollary}

\subsection*{Acknowledgments}
A.G.\ completed this work while a postdoctoral researcher at the University of Oxford and was supported by EPSRC Grant EP/X040461/1 ``Optimisation for Game Theory and Machine Learning''.
E.B., J.F., and R.S.\ were supported by EPSRC Grant EP/W014750/1 ``New Techniques for Resolving Boundary Problems in Total Search''.
We also thank Alexandros Hollender for drawing our attention to the problem addressed in this work.

\bibliography{ref}

@article{deligkas2024pure,
  title={{Pure-Circuit}: Tight Inapproximability for {PPAD}},
  author={Deligkas, Argyrios and Fearnley, John and Hollender, Alexandros and Melissourgos, Themistoklis},
  journal={Journal of the ACM},
  volume={71},
  number={5},
  articleno={31},
  numpages={48},
  year={2024},
  doi={10.1145/3678166}
}

@inproceedings{abbott2005complexity,
  title={On the complexity of two-player win-lose games},
  author={Abbott, Tim and Kane, Daniel and Valiant, Paul},
  booktitle={Proceedings of the 46th Annual {IEEE} Symposium on Foundations of Computer Science ({FOCS} 2005)},
  pages={113--122},
  year={2005},
  publisher={IEEE Computer Society},
  address={Los Alamitos, CA, USA},
  doi={10.1109/SFCS.2005.59}
}

@inproceedings{chen2007approximation,
  title={The approximation complexity of win-lose games},
  author={Chen, Xi and Teng, Shang-Hua and Valiant, Paul},
  booktitle={Proceedings of the 18th Annual {ACM-SIAM} Symposium on Discrete Algorithms ({SODA} 2007)},
  pages={159--168},
  year={2007},
  publisher={Society for Industrial and Applied Mathematics},
  address={Philadelphia, PA, USA}
}

@inproceedings{liu2018approximation,
  title={On the approximation of {N}ash equilibria in sparse win-lose games},
  author={Liu, Zhengyang and Sheng, Ying},
  booktitle={Proceedings of the 32nd {AAAI} Conference on Artificial Intelligence ({AAAI} 2018)},
  pages={1154--1160},
  year={2018},
  publisher={{AAAI} Press},
  address={Palo Alto, CA, USA},
  doi={10.1609/aaai.v32i1.11439}
}

@inproceedings{LiuLD21,
  author       = {Zhengyang Liu and
                  Jiawei Li and
                  Xiaotie Deng},
  title        = {On the Approximation of {N}ash Equilibria in Sparse Win-Lose Multi-player
                  Games},
  booktitle    = {Proceedings of the 35th {AAAI} Conference on Artificial Intelligence ({AAAI} 2021)},
  pages        = {5557--5565},
  year         = {2021},
  publisher    = {{AAAI} Press},
  address      = {Palo Alto, CA, USA},
  doi          = {10.1609/aaai.v35i6.16699},
}

@inproceedings{ChenDT06,
  author       = {Xi Chen and
                  Xiaotie Deng and
                  Shang{-}Hua Teng},
  editor       = {Paul G. Spirakis and
                  Marios Mavronicolas and
                  Spyros C. Kontogiannis},
  title        = {Sparse Games Are Hard},
  booktitle    = {Proceedings of the 2nd International Workshop on Internet and Network Economics ({WINE} 2006)},
  pages        = {262--273},
  year         = {2006},
  publisher    = {Springer},
  address      = {Berlin, Heidelberg},
  series       = {Lecture Notes in Computer Science},
  volume       = {4286},
  doi          = {10.1007/11944874_24},
}

@inproceedings{CodenottiLR06,
  author       = {Bruno Codenotti and
                  Mauro Leoncini and
                  Giovanni Resta},
  editor       = {Yossi Azar and
                  Thomas Erlebach},
  title        = {Efficient Computation of {N}ash Equilibria for Very Sparse Win-Lose
                  Bimatrix Games},
  booktitle    = {Proceedings of the 14th Annual European Symposium on Algorithms ({ESA} 2006)},
  pages        = {232--243},
  year         = {2006},
  publisher    = {Springer},
  address      = {Berlin, Heidelberg},
  series       = {Lecture Notes in Computer Science},
  volume       = {4168},
  doi          = {10.1007/11841036_23},
}

@article{Papadimitriou2023,
  title={Public goods games in directed networks},
  author={Christos Papadimitriou and Binghui Peng},
  journal={Games and Economic Behavior},
  volume={139},
  pages={161--179},
  year={2023},
  doi={10.1016/j.geb.2023.02.002}
}

@inproceedings{Datta2011,
  author={Datta, Samir and Krishnamurthy, Nagarajan},
  editor={Ogihara, Mitsunori and Tarui, Jun},
  title={Some Tractable Win-Lose Games},
  booktitle={Proceedings of the 8th Annual Conference on Theory and Applications of Models of Computation ({TAMC} 2011)},
  series={Lecture Notes in Computer Science},
  volume={6648},
  pages={365--376},
  year={2011},
  publisher={Springer},
  address={Berlin, Heidelberg},
  doi={10.1007/978-3-642-20877-5_36}
}

@incollection{Goldberg2011,
  title={A survey of {PPAD}-completeness for computing {N}ash equilibria},
  author={Goldberg, Paul W.},
  editor={Chapman, Robin},
  booktitle={Surveys in Combinatorics 2011},
  series={London Mathematical Society Lecture Note Series},
  volume={392},
  pages={51--82},
  year={2011},
  publisher={Cambridge University Press},
  address={Cambridge},
  doi={10.1017/CBO9781139004114.003}
}

@misc{Collevecchio2025,
  title={Finding a {N}ash equilibrium of a random win-lose game in expected polynomial time},
  author={Andrea Collevecchio and G{\'a}bor Lugosi and Adrian Vetta and Rui-Ray Zhang},
  year={2025},
  eprint={2510.12846},
  archivePrefix={arXiv},
  primaryClass={cs.GT},
  url={https://arxiv.org/abs/2510.12846}
}

@inproceedings{Hermelin2011,
  author={Hermelin, Danny and Huang, Chien-Chung and Kratsch, Stefan and Wahlstr{\"o}m, Magnus},
  editor={Kolman, Petr and Kratochv{\'i}l, Jan},
  title={Parameterized Two-Player {N}ash Equilibrium},
  booktitle={Proceedings of the 37th International Workshop on Graph-Theoretic Concepts in Computer Science ({WG} 2011)},
  series={Lecture Notes in Computer Science},
  volume={6986},
  pages={215--226},
  year={2011},
  publisher={Springer},
  address={Berlin, Heidelberg},
  doi={10.1007/978-3-642-25870-1_20}
}

@inproceedings{Barman2015,
  title={Approximating {N}ash Equilibria and Dense Bipartite Subgraphs via an Approximate Version of Carath{\'e}odory's Theorem},
  author={Siddharth Barman},
  booktitle={Proceedings of the 47th Annual {ACM} Symposium on Theory of Computing ({STOC} 2015)},
  pages={361--369},
  year={2015},
  publisher={Association for Computing Machinery},
  address={New York, NY, USA},
  doi={10.1145/2746539.2746566}
}

@inproceedings{Kontogiannis2007,
  author={Kontogiannis, Spyros C. and Spirakis, Paul G.},
  editor={Ku{\v{c}}era, Lud{\v{e}}k and Ku{\v{c}}era, Anton{\'i}n},
  title={Well Supported Approximate Equilibria in Bimatrix Games: A Graph Theoretic Approach},
  booktitle={Proceedings of the 32nd International Symposium on Mathematical Foundations of Computer Science ({MFCS} 2007)},
  series={Lecture Notes in Computer Science},
  volume={4708},
  pages={596--608},
  year={2007},
  publisher={Springer},
  address={Berlin, Heidelberg},
  doi={10.1007/978-3-540-74456-6_53}
}

@inproceedings{Bilo2023,
  author={Bil{\`o}, Vittorio and Hansen, Kristoffer Arnsfelt and Mavronicolas, Marios},
  editor={Deligkas, Argyrios and Filos-Ratsikas, Aris},
  title={Computational Complexity of Decision Problems About {N}ash Equilibria in Win-Lose Multi-player Games},
  booktitle={Proceedings of the 16th International Symposium on Algorithmic Game Theory ({SAGT} 2023)},
  series={Lecture Notes in Computer Science},
  volume={14238},
  pages={40--57},
  year={2023},
  publisher={Springer},
  address={Cham, Switzerland},
  doi={10.1007/978-3-031-43254-5_3}
}

@article{Addario007,
  title={A Polynomial Time Algorithm for Finding {N}ash Equilibria in Planar Win-Lose Games},
  author={Addario-Berry, Louigi and Olver, Neil and Vetta, Adrian},
  journal={Journal of Graph Algorithms and Applications},
  volume={11},
  number={1},
  pages={309--319},
  year={2007},
  doi={10.7155/jgaa.00147}
}

@article{DaskalakisGP09,
  author       = {Constantinos Daskalakis and
                  Paul W. Goldberg and
                  Christos H. Papadimitriou},
  title        = {The Complexity of Computing a {N}ash Equilibrium},
  journal      = {{SIAM} Journal on Computing},
  volume       = {39},
  number       = {1},
  pages        = {195--259},
  year         = {2009},
  doi          = {10.1137/070699652},
}

@article{ChenDT09,
  author       = {Xi Chen and
                  Xiaotie Deng and
                  Shang{-}Hua Teng},
  title        = {Settling the complexity of computing two-player {N}ash equilibria},
  journal      = {Journal of the ACM},
  volume       = {56},
  number       = {3},
  articleno    = {14},
  numpages     = {57},
  year         = {2009},
  doi          = {10.1145/1516512.1516516},
}

@article{Brandt2009,
  title={Ranking games},
  author={Felix Brandt and Felix Fischer and Paul Harrenstein and Yoav Shoham},
  journal={Artificial Intelligence},
  volume={173},
  number={2},
  pages={221--239},
  year={2009},
  doi={10.1016/j.artint.2008.10.008}
}

@inproceedings{ghoshUniqueNash,
  title={The Complexity of Computing a Unique {N}ash Equilibrium},
  author={Ghosh, Abheek and Goldberg, Paul W. and Hollender, Alexandros},
  booktitle={Proceedings of the 67th Annual {IEEE} Symposium on Foundations of Computer Science ({FOCS} 2026)},
  year={2026},
  publisher={IEEE Computer Society},
  address={Los Alamitos, CA, USA},
  note={To appear}
}

@inproceedings{Rubinstein16,
  author       = {Aviad Rubinstein},
  title        = {Settling the Complexity of Computing Approximate Two-Player {N}ash Equilibria},
  booktitle    = {Proceedings of the 57th Annual {IEEE} Symposium on Foundations of Computer Science ({FOCS} 2016)},
  pages        = {258--265},
  year         = {2016},
  publisher    = {IEEE Computer Society},
  address      = {Los Alamitos, CA, USA},
  doi          = {10.1109/FOCS.2016.35},
}

@inproceedings{LiptonMM03,
  author       = {Richard J. Lipton and
                  Evangelos Markakis and
                  Aranyak Mehta},
  title        = {Playing large games using simple strategies},
  booktitle    = {Proceedings of the 4th {ACM} Conference on Electronic Commerce ({EC} 2003)},
  pages        = {36--41},
  year         = {2003},
  publisher    = {Association for Computing Machinery},
  address      = {New York, NY, USA},
  doi          = {10.1145/779928.779933},
}

@inproceedings{DaskalakisP09,
  author={Daskalakis, Constantinos and Papadimitriou, Christos H.},
  title={On Oblivious {PTAS}'s for {N}ash Equilibrium},
  booktitle={Proceedings of the 41st Annual {ACM} Symposium on Theory of Computing ({STOC} 2009)},
  pages={75--84},
  year={2009},
  publisher={Association for Computing Machinery},
  address={New York, NY, USA},
  doi={10.1145/1536414.1536427}
}

\appendix
\section{Proof of \cref{thm:pure2poly} -- \PURE{} to $\eps$-\POLY{} Reduction}\label{sec:pure2poly}

We want to prove that the $\eps$-\POLY{} problem, i.e., the problem of computing an $\eps$-WSNE for the class of polymatrix games described in the definition of $\eps$-\POLY{} is PPAD-hard, for all $\eps < 1/2$.
We do this using a reduction from the PPAD-hard problem \PURE{}. We use a variant of \PURE{} with \NOT{}, \AND{}, and \PURIFY{} gates, which is also PPAD-hard. Our reduction adapts a similar reduction by \citet{deligkas2024pure}.  Next we formally define the \PURE{} problem.

\begin{definition}[\PURE{}]
We are given an arithmetic circuit that may have cycles. The circuit has \NOT{}, \AND{}, and \PURIFY{} gates, described below. The input/output variables of each gate take values in the interval $[0,1]$. Each variable must be the output of exactly one gate.
We can also view the circuit as a directed graph with the vertices corresponding to the variables and the directed edges corresponding to the dependency of one variable on another variable through a gate (directed edges point from the inputs to the outputs of the gates).
Next we describe the behaviour of the three gates:
\begin{itemize}
    \itemwithtable{
        \NOT{}: Takes an input, say $u$, and gives an output, say $v$. If $u \in \{0, 1\}$, then $v = 1 - u$. If $u \in (0, 1)$, then $v$ is allowed to take any value in $[0, 1]$.
    }{
        \begin{tabular}{c || c}
            $u$ & $v$ \\
            \hline
            $0$ & $1$ \\
            $1$ & $0$ \\
            $(0,1)$ & $[0, 1]$ \\
        \end{tabular}
    }

    \itemwithtable{
        \AND{}: Takes two inputs, say $u$ and $v$, and gives an output, say $w$. If either input $u$ or $v$ is $0$, then the output $w = 0$. If both inputs $u$ and $v$ are $1$, then the output $w = 1$. Otherwise, the output $w$ is allowed to take any value in $[0, 1]$.
    }{
        \begin{tabular}{c | c || c}
            $u$ & $v$ & $w$ \\
            \hline
            $0$ & $[0, 1]$ & $0$ \\
            $[0, 1]$ & $0$ & $0$ \\
            $1$ & $1$ & $1$ \\
            \multicolumn{2}{c || }{otherwise} & $[0, 1]$ \\
        \end{tabular}
    }

    \itemwithtable{
        \PURIFY{}: Takes one input, say $u$, and gives two outputs, say $v$ and $w$. If the input $u \in \{0, 1\}$, then the outputs $v = w = u$ are the same as the input. If the input $u \in (0, 1)$, then at least one of the output variables must take a value in $\{0, 1\}$, the other output can take any value in $[0, 1]$.
    }{
        \begin{tabular}{c || >{\centering\arraybackslash}p{1cm} | >{\centering\arraybackslash}p{1cm}}
            $u$ & $v$ & $w$ \\
            \hline
            $0$ & $0$ & $0$ \\
            $1$ & $1$ & $1$ \\
            \multirow{2}{*}{$(0, 1)$} & \multicolumn{2}{c}{\multirow{2}{2.8cm}{At least one output in $\{0, 1\}$}} \\
            & \multicolumn{2}{c}{} \\
        \end{tabular}
    }
\end{itemize}
The \PURE{} problem is to find an assignment of all the variables in the range $[0,1]$ such that the variables are consistent with the gates in the circuit. This computational problem is known to be PPAD-hard~\cite{deligkas2024pure}.
Further, this hardness result holds even if we assume that the graph is bipartite and all vertices have an in-degree and out-degree of at most two and degree (= in-degree + out-degree) of at most three~\cite{deligkas2024pure}.
We slightly overload notation and identify both a vertex (variable) and its value by $v$.
Let us denote the set of vertices as $V$.
\end{definition}

\paragraph{Reduction.} Given an instance of the \PURE{} problem, we construct a polymatrix game with $|V|$ players satisfying the requirements stated in the definition of the $\eps$-\POLY{} problem.
Each player in the polymatrix game has two actions $\{ 0, 1 \}$ and simulates a variable from the \PURE{}. In fact, the mixed strategy of any given player will directly correspond to the corresponding variable in the \PURE{}. Therefore, overloading notation,
we identify the players of the polymatrix game as elements from $V$,
and denote by $v \in V$ all of the following: (i) the variable $v$ in the \PURE{}; (ii) the corresponding player of the polymatrix game; and (iii) the probability with which this player plays action $1$.

The payoff matrices of the game will simulate the gates. The definition of \PURE{} requires that each variable $v$ is the output of exactly one gate $g$. An important property of our reduction is that the player representing $v$ only receives non-zero payoffs from the inputs to the gate $g$, and receives zero payoffs from any other gate that use $v$ as an input. This means that we can reason about the equilibrium condition of each of the gates independently, without worrying about where the values computed by those gates are used.

\NOT{}. The player $v$, who represents the output variable, will
play the following bimatrix game against player $u$, who represents the input variable.
\[
    A^{vu} =
        \begin{pmatrix}
            0 & 1\\
            1 & 0
        \end{pmatrix},
    \qquad
    A^{uv} = 0.
\]
We claim that this gate will work for all $\eps < 1$.
\begin{itemize}
    \item If $u$ plays action $0$ as a pure strategy, then playing $1$ gives a payoff of $1$ to $v$ and playing $0$ gives a payoff $0$ to $v$. So in any $\eps$-WSNE with $\eps < 1$, only action $1$ can be an $\eps$-best-response for $v$, implying that $v$ must play action $1$ as a pure strategy, as required by the constraints of the \NOT{} gate.

    \item Using identical reasoning, if $u$ plays action $1$ as a pure strategy, then $v$ must play action $0$ as a pure strategy in any $\eps$-WSNE.

    \item If $u$ plays a strictly mixed strategy, then the \NOT{} gate places no constraints on the output, so we need not prove anything about $v$’s strategy.
\end{itemize}

\AND{}. The player $w$, who represents the output variable, will play the following bimatrix games against $u$ and $v$, who represent the input variables.
\[
    A^{wu} = A^{wv} =
        \begin{pmatrix}
            2 & 0\\
            0 & 1
        \end{pmatrix},
    \qquad
    A^{uw} = A^{vw} = 0.
\]
We claim that this gate will work for all $\eps < 1$.
\begin{itemize}
    \item If both $u$ and $v$ play action $1$ as a pure strategy, then the payoff to $w$ for playing action $0$ is $0$, and the payoff to $w$ for playing action $1$ is $1 + 1 = 2$. So if $\eps < 2$, then in all $\eps$-WSNEs, the only $\eps$-best-response for $w$ is action $1$, as required by the constraints of the \AND{} gate.

    \item If at least one of $u$ or $v$ plays action $0$ as a pure strategy, then the payoff to $w$ for playing action $0$ is at least $2$, while the payoff to $w$ for playing action $1$ is at most $1$. So if $\eps < 1$, then in all $\eps$-WSNEs the only $\eps$-best-response for $w$ is action $0$, as required by the constraints of the \AND{} gate.

    \item The \AND{} gate places no other constraints on the variable $w$, so we can ignore all other cases, e.g., the case where both $u$ and $v$ play strictly mixed strategies.
\end{itemize}

\PURIFY{}. The players $v$ and $w$, who represent the output variables, play the following games against $u$, who represents the input variable.
\[
    A^{vu} =
        \begin{pmatrix}
            1 & 0\\
            0 & 2
        \end{pmatrix},
    \qquad
    A^{wu} =
        \begin{pmatrix}
            2 & 0\\
            0 & 1
        \end{pmatrix},
    \qquad
    A^{uv} = A^{uw} = 0.
\]
We claim that this gate works for any $\eps < 1/2$. We first consider player $v$.
\begin{itemize}
    \item If $u$ plays action $0$ as a pure strategy, then action $0$ gives payoff $1$ to $v$, while action $1$ gives payoff $0$, so in any $\eps$-WSNE with $\eps < 1$, we have that action $0$ is the only $\eps$-best-response for $v$.

    \item If $u$ places at least $1/2$ probability on action $1$, then action $0$ gives at most $1/2$ payoff to $v$, while action $1$ gives at least $1$ payoff to $v$. So in any $\eps$-WSNE with $\eps < 1/2$, we have that action $1$ is the only $\eps$-best-response for $v$.
\end{itemize}
Symmetrically, for player $w$ we have the following:
\begin{itemize}
    \item If $u$ places at most $1/2$ probability on action $1$, then action $0$ gives at least $1$ payoff to $w$, while action $1$ gives at most $1/2$ payoff to $w$. So in any $\eps$-WSNE with $\eps < 1/2$, we have that action $0$ is the only $\eps$-best-response for $w$.

    \item If $u$ plays action $1$ as a pure strategy, then action $1$ gives payoff $1$ to $w$, while action $0$ gives payoff $0$, so in any $\eps$-WSNE with $\eps < 1$ we have that action $1$ is the only $\eps$-best-response for $w$.
\end{itemize}
From these properties, we can verify that the constraints imposed by the \PURIFY{} gate are enforced correctly.
\begin{itemize}
    \item If $u$ plays action $0$ as a pure strategy, then both $v$ and $w$ play action $0$ as a pure strategy.
    \item If $u$ plays action $1$ as a pure strategy, then both $v$ and $w$ play action $1$ as a pure strategy.

    \item No matter how $u$ mixes, at least one of $v$ or $w$ plays a pure strategy, with $w$ playing pure strategy $0$ whenever $u$ places at most $1/2$ probability on action $1$, and $v$ playing pure strategy $1$ whenever $u$ places at least $1/2$ probability on action $1$.
\end{itemize}

As we have argued above, each of the gate constraints from the \PURE{} instance are enforced correctly in any $\eps$-WSNE of the polymatrix game with $\eps < 1/2$. Thus, given such an $\eps$-WSNE, we get a satisfying assignment to the \PURE{} instance. The interaction graph of the polymatrix game is exactly the interaction graph of the \PURE{} instance, so we get bipartiteness, and we get in-degree, out-degree, and degree of at most $2$, $2$, and $3$, respectively. Moreover, by construction, the game has two strategies for each player and all the payoffs matrices are either zero or among one of these three matrices: $\begin{pmatrix} 0 & 1 \\ 1 & 0 \end{pmatrix}$, $\begin{pmatrix} 1 & 0 \\ 0 & 2 \end{pmatrix}$, or $\begin{pmatrix} 2 & 0 \\ 0 & 1 \end{pmatrix}$. Finally, if the out-degree of a player $u$ is two, then we could have one of the following two scenarios:
\begin{itemize}
    \item If $u$ is an input to a \PURIFY{} gate, then the two outgoing non-zero payoff matrices are $\begin{pmatrix} 1 & 0 \\ 0 & 2 \end{pmatrix}$ and $\begin{pmatrix} 2 & 0 \\ 0 & 1 \end{pmatrix}$, which satisfies the requirements of $\eps$-\POLY{}.

    \item Otherwise, $u$ is an input to two gates that could be any combination of \NOT{} and \AND{} gates. If $u$ is an input to an \AND{} gate, e.g., $w = \AND(u, v)$ for some $v$ and $w$, then we pass $u$ through a sequence of two \NOT{} gates before putting it into the \AND{} gate: $w = \AND(\NOT(\NOT(u)), v)$. Notice that this does not change the behaviour of the circuit. By doing this we ensure that if $u$ has an out-degree of two, then $u$ cannot be an input to an \AND{} gate. Hence, the only case other than the \PURIFY{} case discussed in the previous point is that $u$ is an input to two \NOT{} gates and the two outgoing non-zero payoff matrices are both $\begin{pmatrix} 0 & 1 \\ 1 & 0 \end{pmatrix}$, which satisfies the requirements of $\eps$-\POLY{}.
\end{itemize}
Hence, the constructed game satisfies all the requirements of $\eps$-\POLY{}.

\section{Proofs of Lemmas from \cref{thm:res-bimatrix}}

\subsection{Checking the Properties of the Constructed $\{ 0, 1, 2, 8 \}$-Valued Bimatrix Game}
\label{app:resbi-check}
We check that the constructed bimatrix game $(A, B)$ satisfies the requirements of $\eps$-\RESBI{m}:
The input polymatrix game $G$ has in-degree and out-degree of two, all payoffs are in the set $\{ 0, 1, 2 \}$, and the payoff matrices are diagonal or anti-diagonal. Therefore, in the top-left primary-primary block of size $2n \times 2n$ of the matrix $A$, every row and column has at most two non-zero payoffs and the non-zero payoffs are in the set $\{ 1, 2 \}$.
For the other three blocks of $A$ that include secondary row or column indices, by construction:
\begin{itemize}
    \item the primary columns get a $K$ from the secondary rows, so these columns total have one $K$ and at most two payoffs form $\{ 1, 2 \}$;
    \item the secondary columns either have two $1$s from the primary rows or have one $1$ from the secondary rows, no other non-zero payoffs;
    \item the primary rows get a $1$ from the secondary columns, so these rows total have at most three payoffs form $\{ 1, 2 \}$;
    \item the secondary rows either have two $K$s from the primary columns or have one $1$ from the secondary columns, no other non-zero payoffs.
\end{itemize}
Overall, the above properties show that the rows and columns of $A$ satisfy the requirements $\eps$-\RESBI{m}. A similar argument applies for $B^\intercal$.

\subsection{Proof of \cref{lm:positive-prob}} \label{app:positive-prob-proof}
Define two bijections $f,g : [N] \to [N]$ by
\[
    f(i) = \bmod(i + n + 2, N),
    \qquad
    g(j) = \bmod(j + n + 1, N).
\]

We first show that if $x_i < 1/2KN$ for some $i \in [N]$, then $y_{f(i)} = 0 < 1/2KN$.
Fix $i$ with $x_i < 1/2KN$ and consider the column index $f(i)$. We have the following two cases:

\smallskip
\noindent\textbf{Case 1: $f(i) \in [N] \setminus [n]$ (a secondary column index).}
By construction, the column player's payoff for playing secondary action $f(i)$ is positive only when the row player plays index $i$ -- it is $K$ if $i \in [n]$ and $1$ if $i \in [N] \setminus [n]$. So, if $x_i < 1/2KN$, the column player gets less than $K \cdot (1/2KN) = 1/2N$ expected utility for playing $f(i)$.

On the other hand, some row index has probability at least $1/N$ under $x$.
Against any fixed row index, the column player has a pure action with payoff at least $1$.
Thus the best-response payoff of the column player is at least $1/N$.
Since $\eps < 1/6KN < 1/2N$, an action with payoff $1/2N$ cannot be an $\eps$-best-response, hence it cannot appear in the support of $y$, and therefore $y_{f(i)} = 0$.

\smallskip
\noindent\textbf{Case 2: $f(i) \in [n]$ (a primary column index).}
Let $j = f(i) \in [n]$.
If the column player plays a pure action $(j, t)$ for some $t \in \{0,1\}$, then because $x_i < 1/2KN$, the column player's expected utility receives a contribution of less than $1/2KN$ from the unique secondary row index that pays $1$ against column index $j$ (this is exactly the index $i$ satisfying $\bmod(j-i, N) = n + 2$).
So the expected utility of action $(j,t)$ must come mainly from the top-left primary-primary block.

Now, in the polymatrix instance, the in-degree of the player corresponding to $j$ is at most two.
By the definition of $\delta$-\POLY{}, the (at most) two non-zero incoming payoff matrices for this player are either diagonal or anti-diagonal matrices with payoffs in $\{ 0, 1, 2 \}$.
Let the set of primary row indices with non-zero payoffs corresponding to the column index $j$ be $S$. We know that $|S| \le 2$. Also let $k^* = \argmax_{k \in S} x_k$.
The column player's utility for playing index $j$ is less than
\begin{equation}\label{eq:positive-prob:1}
    \frac{1}{2KN} + 2 \sum_{k \in S} x_k \le \frac{1}{2KN} + 2 |S| \max_{k \in S} x_k \le \frac{1}{2KN} +4 x_{k^*}.
\end{equation}
Consider the secondary column action $f(k^*) = k^* + n + 2$ corresponding to the primary row index $k^*$. Its expected utility is exactly $K x_{k^*} = 8 x_{k^*}$. Hence the difference in utility between $f(k^*)$ and $(j,t)$ is at least $8 x_{k^*} - 4 x_{k^*} - 1/2KN = 4 x_{k^*} - 1/2KN$.
If $4 x_{k^*} - 1/2KN > \eps$, then $(j,t)$ cannot be an $\eps$-best-response, so it cannot be in the support.
If $4 x_{k^*} - 1/2KN \le \eps$, then the utility for playing $(j,t)$ is at most $4 x_{k^*} + 1/2KN \le \eps + 1/KN$ by \eqref{eq:positive-prob:1}. But, as in Case 1, the column player's best-response payoff is at least $1/N$.
Since $\eps < 1/6KN$, we have $1/N - (\eps + 1/KN) > \eps$, and therefore $(j,t)$ again cannot be an $\eps$-best-response.
Thus no pure action in index $j$ can appear in the support, which implies $y_{f(i)}=0$.

\smallskip
We have proved that $x_i < 1/2KN \implies y_{f(i)} = 0 < 1/2KN$.
By symmetry (swapping the roles of the two players and using the definition of $g$ instead of $f$), we also get $y_j < 1/2KN \implies x_{g(j)} = 0 < 1/2KN$.

Finally, note that $g(f(i)) = \bmod(i+1, N)$ for every $i \in [N]$.
Therefore, if any $x_i < 1/2KN$ then $x_{i+1} < 1/2KN$, then $x_{i+2 } < 1/2KN$, and so on, which implies $\sum_{i \in [N]} x_i < 1/2K < 1$, a contradiction.
Hence $x_i \ge 1/2KN$ for all $i \in [N]$, and then also $y_j \ge 1/2KN$ for all $j \in [N]$ by using $y_j < 1/2KN \implies x_{g(j)} < 1/2KN$.

\subsection{Proof of \cref{lm:unif-over-indices}} \label{app:unif-over-indices-proof}
By \cref{lm:positive-prob}, every secondary action is in the support.
Fix $i \in [n]$ and consider the secondary column action $i+n+2$.
By construction, its expected utility is exactly $K x_i$ (it only pays $K$ against primary index $i$ and has no other positive contributions).
Since all supported actions of the column player must be $\eps$-best-responses, the payoffs of $i+n+2$ and $j+n+2$ can differ by at most $\eps$ for any $i,j \in [n]$.
Thus $|K x_i - K x_j| \le \eps$, i.e., $|x_i-x_j| \le \eps/K$.

The proof for $|y_i-y_j|$ is symmetric using the secondary row actions $i+n+1$ and $j+n+1$, whose expected utilities are exactly $Ky_i$ and $Ky_j$.

\subsection{Proof of \cref{lm:encode}} \label{app:encode-proof}
We prove the claim for the players on the left side of $G$ (corresponding to the row player’s primary indices); the proof for the right side is symmetric.

We relate the payoff gaps between the two actions $(i, 0)$ and $(i, 1)$ for any given primary index $i$ in $(A,B)$ with the payoff gaps between the actions $0$ and $1$ for player $i$ in $G$.

Fix $i \in [n]$. For $s \in \{0, 1\}$, denote the expected utility of player $i$ in the polymatrix game $G$ when she plays $s$ against $q$ by
\[
    u_i(s, q) :=\sum_{j \in [n]} \big( (1-q_j)\, A^{ij}_{s0} + q_j\, A^{ij}_{s1} \big).
\]
By the in-degree condition of $\delta$-\POLY{}, there are at most two indices $j$ for which $A^{ij} \neq 0$; let
\[
    T_i :=\{ j \in [n] \mid A^{ij} \neq 0 \},
    \qquad \text{so } |T_i| \le 2.
\]
For $j \in [n]$, define
\[
    \Delta^{ij} :=\big((1-q_j)\,A^{ij}_{0 0} + q_j\,A^{ij}_{0 1}\big) \ -\ \big((1-q_j)\,A^{ij}_{1 0} + q_j\,A^{ij}_{1 1}\big).
\]
Since all entries of $A^{ij}$ lie in $\{0,1,2\}$, we have $|\Delta^{ij}| \le 2$ for every $j$.

Let $u_A((i, s), y)$ denote the expected utility of the row player for playing primary action $(i, s)$ against mixed strategy $y$ of the secondary player.
Consider the two row actions $(i,0)$ and $(i,1)$ in the bimatrix game.
Their payoffs against secondary column actions are identical (the secondary payoffs depend only on the index $i$ and not on the bit),
so the payoff difference comes only from the primary-primary block.
Therefore,
\begin{equation}\label{eq:encode:gap-bimatrix}
    u_A((i,0), y) - u_A((i,1), y)
    = \sum_{j \in [n]} y_j \, \Delta^{ij}
    = \sum_{j \in T_i} y_j \, \Delta^{ij}.
\end{equation}
On the other hand, the corresponding payoff gap in the polymatrix game is
\begin{equation}\label{eq:encode:gap-poly}
    u_i(0, q) - u_i(1, q)
    = \sum_{j \in T_i} \Delta^{ij}.
\end{equation}
Let $y_{\min} := \min_{j \in [n]} y_j$.
By \cref{lm:positive-prob}, $y_{\min} \ge 1/2KN$.
By \cref{lm:unif-over-indices}, we have
$|y_j - y_{j'}| \le \eps/K$ for all $j, j' \in [n]$, therefore
$0 \le y_j - y_{\min} \le \eps/K$ for all $j \in [n]$.
Using \eqref{eq:encode:gap-bimatrix} and $|\Delta^{ij}|\le 2$, we get
\begin{align*}
    u_A((i,0), y) - u_A((i,1), y)
    &= \sum_{j \in T_i} y_j \, \Delta^{ij} \\
    &= y_{\min} \sum_{j \in T_i} \Delta^{ij} - \sum_{j \in T_i} (y_j-y_{\min}) \, \Delta^{ij} \\
    &\ge \frac{1}{2 K N} \sum_{j \in T_i} \Delta^{ij} - \sum_{j \in T_i} (y_j-y_{\min})\,|\Delta^{ij}| \\
    &\ge \frac{1}{2 K N} \sum_{j \in T_i} \Delta^{ij} - |T_i| \cdot \frac{\eps}{K} \cdot 2 \\
    &\ge \frac{1}{2 K N} (u_i(0, q) - u_i(1, q)) - \frac{\eps}{2},
\end{align*}
where the last inequality is by \eqref{eq:encode:gap-poly} and $|T_i| \le 2$.
Say $x_{i1} > 0$, then $(i, 1)$ must be an $\eps$-best-response for the row player, so
\begin{align*}
    & u_A((i,1), y) \ge u_A((i,0), y) - \eps \implies u_A((i,0), y) - u_A((i,1), y) \le \eps \\
    &\implies \frac{1}{2 K N} (u_i(0, q) - u_i(1, q)) - \frac{\eps}{2} \le \eps \implies u_i(0, q) - u_i(1, q) \le 3 K N \eps = \delta.
\end{align*}
Therefore, action $1$ is a $\delta$-best-response for the polymatrix player $i$, and $p_i = x_{i1}/x_i > 0$ is justified. A similar argument applies when $x_{i0} > 0  \iff 1 - p_i > 0$ (and for the right side players of the polymatrix game). Therefore, $(p,q)$ is a $\delta$-WSNE of the polymatrix game $G$ for $\delta = 3 K N \eps$.

\section{Formal Definitions for the Simulation Matrices}
\label{app:matrices}

\paragraph{\bf Formal definition of the matrices used in a single
column simulations.}
Then the matrices $A'$ and $B'$ can be defined in the following way.
\begin{itemize}
\item For each column $j < k$ of $A$, the block of $A'$ defined by $\cb_1^j$ and
$\rb^j$ is equal to $S$, while the corresponding block of $B'$ is equal to $T$.

\item For each column $j < k$ of $A$, the block of $A'$ defined by $\cb_2^j$ and
$\rb^j$ is equal to a $2 \times 1$ vector filled with $1$s.

\item For each column $j < k$ of $A$, the block of $A'$ defined by $\cb_1^j$ and
$\rb_e$ is equal to $\encode(A_j)$.

\item For each column $j < k$ of $A$, the block of $B'$ defined by $\cb_2^j$ and
$\rb_e$ is equal to $B_j$.

\item For each column $j \ge k$ of $A$, the block of $A'$ defined by $\rb_e$
and column $3k + (j - k)$ is equal to $A_{j}$, and the corresponding
block of $B'$ is equal to $B_{j}$.

\item All other entries of both matrices are equal to zero.

\end{itemize}

\paragraph{\bf Formal definition of the matrices used in a dual column
simulation}
\begin{itemize}
\item For each row $i < k$, the block of $A'$ defined by $\rb_1^i$ and
$\cb_1^i$ is defined to be $S$, and the corresponding block of $B'$ is defined
to be $T_1$.
\item For each row $i < k$, the block of $A'$ defined by $\rb_2^i$ and
$\cb_1^i$ is defined to be $S$, and the corresponding block of $B'$ is defined
to be $T_2$.
\item For each row $i < k$, the block of $A'$ defined by $\rb_1^i$ and
$\cb_2^i$ is defined to be an all $1$s matrix. Likewise, the block of $A'$
defined by $\rb_2^i$ and $\cb_3^i$ is defined to be an all $1$s matrix.

\item For each row $i < k$, the block of $A'$ defined by $\rb_e$ and $\cb_1^i$
is the output of applying the $\encode$ function to columns $2i$ and $2i + 1$
of $A$.

\item For each row $i < k$, the block of $B'$ defined by $\rb_e$ and $\cb_2^i$
is defined to be equal to $B_{2i}$, and the
block of $B'$ defined by $\rb_e$ and $\cb_3^i$ is defined to be equal to $B_{2i
+ 1}$.

\item For each row $i \ge k$, the block of $A'$ defined by $\rb_e$ and column
$5k + i$ is the $2k+i$th column of $A$, while the corresponding block of
$B'$ is the $2k+i$th column of $B$.
\end{itemize}

\section{Proof of Theorem~\ref{thm:scs1}}
\label{app:scs1proof}

Throughout this section we assume that
$(x^{*}, y^{*})$ is an $\eps$-WSNE of $(A', B')$. For each set of columns
$C$, we use $P(C) = \sum_{j \in C} y^*_j$ to give the total amount of
probability placed on those columns by $y^*$, and likewise for each set of rows
$R$ we use $P(R) = \sum_{j \in R} x^*_j$ to give the total amount of
probability given to those rows by $x^*$.

We start by showing some basic properties about $(x^*, y^*)$. The following
lemma states that if the column player plays column block $\cb_1^j$ for some $j
< k$, then the row player must play the corresponding row block.

\begin{lemma}
\label{lem:scs1}
If $\eps < 1/(3n)$, then for each $j < k$, if $P(\cb_1^j) > 0$ holds in $(x^{*}, y^{*})$, then $P(\rb^j) > 0$ also holds.
\end{lemma}
\begin{proof}
Suppose for the sake of contradiction that this is not the case, meaning
that $P(\rb^j) = 0$. This implies that the column player receives payoff $0$
for all columns in $\cb_1^j$, since all other rows in $\cb_1^j$ have zero
payoff.
Observe that every column of $B'$ contains at least one payoff that is greater
than or equal to $1$ by construction, so there must exist a column for which
the column player receives payoff at least $1/(3n)$.
Since $\eps < 1/(3n)$ by assumption, we have that the column player cannot
place probability on any column in $\cb_1^j$, which provides our contradiction.
\end{proof}

Observe that the unique Nash equilibrium of $(S, T)$ has the column player
playing $(1/3, 2/3)$ over the two columns.
The next lemma states that if the row player plays row block $\rb_j$ for some
$j < k$, then the column player must play a strategy that is within $\eps$ of
a scaled-down version of this strategy in the corresponding column block.

\begin{lemma}
\label{lem:colprob}
If $\eps < 1/(3n)$, and if $P(\rb_j) > 0$, then we have $y^*_{3j} \in [1/3
\cdot P(\cb_1^j) - \eps, 1/3 \cdot P(\cb_1^j) + \eps]$.
\end{lemma}
\begin{proof}
First suppose for the sake of contradiction that $y^*_{3j} < 1/3 \cdot P(\cb_1^j) - \eps$. Then the payoff
to the row player of row $2j$ is at most $2/3 \cdot P(\cb_1^j) - 2\eps +
y^*_{3j + 2}$. On the other hand, we have that $y^*_{3j + 1} =
P(\cb_1^j) - y^*_{3j}$, and therefore
\begin{align*}
y^*_{3j + 1} &> P(\cb_1^j) - (1/3 \cdot P(\cb_1^j) - \eps) \\
&= 2/3 \cdot P(\cb_1^j) + \eps.
\end{align*}
Thus the payoff to the row player of row $2j + 1$ is at least
$2/3 \cdot P(\cb_1^j) + \eps + y^*_{3j + 2}$. So the difference in payoff between rows $2j + 1$ and $2j$ is at least
\begin{align*}
& 2/3 \cdot P(\cb_1^j) + \eps + y^*_{3j + 2} - 2/3 \cdot P(\cb_1^j) - 2\eps -
y^*_{3j + 2} \\
& = 3 \eps,
\end{align*}
meaning that the row player cannot play row $2j$ in an $\eps$-WSNE. But this
means that
the column player must receive payoff zero for column $3j + 1$. Since every
row of $B'$ contains a payoff entry that is at least $1$, there must exist a
column for which the column player receives payoff at least $1/(3n)$. Since
$\eps < 1 / (3n)$, this means that the column player must place zero
probability on column $3j + 1$, which contradicts the assumption that
$y^*_{3j} < 1/3 \cdot P(\cb_1^j) - \eps$.

The fact that $y^*_{3j} \le 1/3 \cdot P(\cb_1^j) + \eps$ can be proved in a
symmetric manner.
\end{proof}

Recall that we exclude all column blocks that are played with probability less
than $11 \eps$ when we construct $y$ from $y^*$. The following lemma states
that, for any column block $\cb_1^j$ that is not excluded in this way, the
column player must also play $\cb^j_2$ as well.

\begin{lemma}
\label{lem:cb2lower}
If $P(\rb^j) > 0$ and $P(\cb^j_1) \ge 11 \eps$ hold in $(x^*, y^*)$ for
some column $j < k$, and if $4 \le K \le 8$ and $\eps < 1/(3n)$, then $P(\cb^j_2) > 0$ also holds.
\end{lemma}
\begin{proof}
Let us assume for the sake of contradiction that
$P(\cb^j_2) = 0$.
We can invoke Lemma~\ref{lem:colprob} to conclude that $y^*_{3j} \le 1/3 \cdot
P(\cb_1^j) + \eps$, and therefore $y^*_{3j+1} \ge 2/3 \cdot P(\cb_1^j) - \eps$.
Lemma~\ref{lem:colprob} also implies that
$y^*_{3j+1} \le 2/3 \cdot P(\cb_1^j) + \eps$.

Since $P(\cb^j_2) = 0$ by assumption, the payoff of rows $2j$ and $2j + 1$ can
therefore be upper bounded by
\begin{align*}
(A' \cdot y^*)_{2j} &\le 2 \cdot (1/3 \cdot P(\cb_1^j) + \eps), \\
(A' \cdot y^*)_{2j + 1} &\le 1 \cdot (2/3 \cdot P(\cb_1^j) + \eps).
\end{align*}
Therefore the payoff of both rows is bounded by
$2/3 \cdot P(\cb_1^j) + 2\eps$.

On the other hand, observe that by construction the $\encode$ function always
outputs a value of~$K/2$ in column $3j + 1$, and let $i \in \rb_e$ be the row
that contains this payoff. We have that the payoff of this row is at least
\begin{equation*}
K/2 \cdot (2/3 \cdot P(\cb_1^j) - \eps)
\end{equation*}
So the difference in payoff between row $i$ and either of the rows $2j$ or $2j
+ 1$ is at least
\begin{align*}
& K/2 \cdot (2/3 \cdot P(\cb_1^j) - \eps) - 2/3 \cdot P(\cb_1^j) - 2\eps \\
& = (K/2 - 1) (2/3 \cdot P(\cb_1^j)) - (2 + K/2) \cdot \eps \\
& \ge (K/2 - 1) (2/3 \cdot P(\cb_1^j)) -6  \cdot \eps \\
& \ge 2/3 \cdot P(\cb_1^j) -6  \cdot \eps \\
& > \eps
\end{align*}
where the first inequality uses the assumption that $K \le 8$, the second
inequality uses the assumption that $K \ge 4$, and the final inequality uses
the assumption that $P(\cb_1^j) \ge 11 \eps$.
Hence, the row player cannot place probability on either row $2j$ or $2j + 1$
in an $\eps$-WSNE, contradicting the assumption that $P(\rb^j) > 0$.
\end{proof}

We now turn our attention towards showing that the re-normalization that is
used when defining $x$ is well-defined. The first lemma shows that if the
column player plays $\cb_2^j$ for some $j < k$, then the row player must place
some probability on $\rb_e$, which means that the re-normalization will
succeed.

\begin{lemma}
\label{lem:rbelowera}
If $P(\cb_2^j) > 0$ holds for some column $j < k$ in $(x^*, y^*)$, and if $K
\le 8$ and $\eps < 1/(6n)$, then $P(\rb_e) \ge 1/(48n)$ also holds.
\end{lemma}
\begin{proof}
Suppose for the sake of contradiction that $P(\rb_e) < 1/(48n)$. Since each row
of $B'$ contains a payoff entry that is at least $1$, there must exist a column
for which the column player receives payoff at least $1/(3n)$. On the other
hand, in column $3j+2$, which is the defining column of $\cb_2^j$, all payoff
entries are zero outside of $\rb_e$, and the maximum payoff inside $\rb_e$ is
$K \le 8$. Thus, if $P(\rb_e) < 1/(48n)$, the payoff of column $3i+2$ is
strictly less than $8 \cdot 1/(48n) = 1/(6n)$. So the difference in payoff
between column $3i+2$ and a best-response column is strictly larger than
$1/(3n) - 1/(6n) = 1/(6n)$. Since $\eps < 1/(6n)$ by assumption, we can
conclude that the column player cannot play column $3i+2$ in an $\eps$-WSNE,
which contradicts the assumption that $P(\cb_2^j) > 0$.
\end{proof}

While the previous lemma dealt only with one case, the next lemma shows that
the re-normalization that produces $x$ also works in all other cases.

\begin{lemma}
\label{lem:rbelower}
If $4 \le K \le 8$ and $\eps < 1/(11n)$, then $P(\rb_e) \ge 1/(48n)$ holds in $x^*$.
\end{lemma}
\begin{proof}
Suppose for the sake of contradiction that $P(\rb_e) < 1/(48n)$. Then columns $j \ge
3k$ give payoff at most $K /(48n) \le 1/(6n)$
to the column player. Since each row of $B'$ has a payoff
entry that is at least 1, the column player's best-response payoff is at least
$1/(3n)$, and since $\eps < 1/(6n)$ by assumption, we have that the column
player cannot play any column $j \ge 3k$ in an $\eps$-WSNE.

Next note that if there exists a column $j < k$ such that $P(\cb_2^j) > 0$,
then Lemma~\ref{lem:rbelowera} implies that $P(\rb_e) \ge 1/(48n)$, which provides a
contradiction.

So we can now proceed assuming that $P(\cb_2^j) = 0$ for all $j < k$. This
implies that all probability assigned by $y^*$ must be to columns that lie in
$P(\cb_1^j)$ for some $j < k$. Hence, there must exist a column $j$ such that
$P(\cb_1^j) \ge 1/k \ge 1/n$. Since $\eps < 1/(11n)$ by assumption we therefore
have $P(\cb_1^j) > 11 \eps$. Lemma~\ref{lem:scs1} then implies that $P(\rb^j)
> 0$ holds, so we meet the precondition of Lemma~\ref{lem:cb2lower}, which
implies that $P(\cb_2^j) > 0$. But this contradicts the fact that we are in
the case where $P(\cb_2^j) = 0$ for all $j$.
\end{proof}

Next we show that the re-normalization that is used to define $y$ is also well
defined. We begin by ignoring the step that excludes column blocks whenever their
probability is below $11 \eps$. The following lemma states that there is enough
probability on the columns that define $y$ before we exclude those columns.

\begin{lemma}
\label{lem:clower}
If $4 \le K \le 8$ and $\eps < 1/(11n)$ then
$\sum_{j < k} P(\cb_1^j) + \sum_{j \ge k} y^*_{3k + j} \ge 1/(24n)$ holds in~$y^*$.
\end{lemma}
\begin{proof}
Suppose for the sake of contradiction that
$\sum_{j < k} P(\cb_1^j) + \sum_{j \ge k} y^*_{3k + j} < 1/(24n)$.
Note that the rows in $\rb_e$ have non-zero payoffs only in columns that are in
$\cb_1^j$ for some $j < k$ or in columns $j \ge 3k$, and moreover, all payoff
entries in these rows are bounded by $K/2$.
This implies that each row $i \in \rb_e$ has payoff at most $K/2 \cdot (1/24n)
\le 1/(6n)$. On the other hand, note that each column of $A'$ contains at least
one payoff entry that is at least 1, so the row player's best-response payoff
must be at least $1/(3n)$. Since $\eps < 1/(6n)$ we can conclude that the row
player cannot play any row in $\rb_e$ in an $\eps$-WSNE, but this directly
contradicts the conclusion of Lemma~\ref{lem:rbelower}.
\end{proof}

Now we take into account the fact that some column blocks are excluded when we
define $y$, and the following lemma states that the amount of probability lost
through excluding those columns is not large enough to cause the re-normalization
step used to define $y$ to fail.

\begin{lemma}
\label{lem:clowerfinal}
Suppose $4 \le K \le 8$ and $\eps < 1/(240 \cdot n^2)$ then
if $C = \{j < k \; : \; P(\cb_1^j) \ge 11 \eps \}$ we have
$\sum_{j \in C} P(\cb_1^j) + \sum_{j \ge k} y^*_{3k + j} \ge 11/(240 n)$ holds in~$y^*$.
\end{lemma}
\begin{proof}
By Lemma~\ref{lem:clower}, we have that
$\sum_{j < k} P(\cb_1^j) + \sum_{j \ge k} y^*_{3k + j} \ge 1/(24n)$.
At most $n \cdot 11 \eps$ probability can be allocated to column blocks
$\cb_1^j$ that satisfy
$P(\cb_1^j) < 11 \eps$.
Since $\eps < 1/(240 \cdot n^2)$, we therefore have that at most $11/(240 n)$
probability is allocated to those columns. Hence, we have
$\sum_{j \in C} P(\cb_1^j) + \sum_{j \ge k} y^*_{3k + j} \ge 11/(240 n)$.
\end{proof}

So we have now shown that both $x$ and $y$ are well-defined. We can now turn
our attention towards showing the remaining claim of Theorem~\ref{thm:scs1}
holds. The following lemma shows that the row player's strategy $x$ satisfies
the requirements of a $O(n \cdot \eps)$-WSNE in $(A, B)$.

\begin{lemma}
\label{lem:scsrow}
If $4 \le K \le 8$ and $\eps < 1/(11n)$, then
$\left|x^\intercal \cdot B_j - x^\intercal \cdot B_{j'} \right| \le 48 \cdot n \cdot \eps.$
for all columns $j,j'$ for which $y_j > 0$ and $y_{j'} > 0$.
\end{lemma}
\begin{proof}
Since $(x^*, y^*)$ is an $\eps$-WSNE of $(A',
B')$, for each pair of columns $j, j'$ with $y^*_j > 0$ and $y^*_{j'} > 0$ we
have that $\left|(x^*)^\intercal \cdot B'_j - (x^*)^\intercal \cdot B'_{j'} \right| \le \eps$.

For each column $j \le k$, observe that we have $y_j > 0$ only if $P(\cb_1^j) >
11 \eps$
by definition, hence we can apply Lemmas~\ref{lem:scs1} and~\ref{lem:cb2lower}
to conclude that $P(\cb_2^j) > 0$. Observe that the payoff of column $3j + 2$,
which is the only column in $\cb_2^j$ is given by $P(\rb_e) \cdot x^\intercal \cdot
B_j$ by construction. Likewise, for any column $j \ge k$, the payoff of column $3k + j$ is
given by $P(\rb_e) \cdot x^\intercal \cdot B_j$ by construction.

So for each pair of columns $j, j'$ for which $y_j > 0$ and $y_{j'} > 0$, we have
$$\left|P(\rb_e) \cdot x^\intercal \cdot B_j - P(\rb_e) \cdot x^\intercal \cdot B_{j'} \right|
\le \eps,$$
which implies
$$\left|x^\intercal \cdot B_j - x^\intercal \cdot B_{j'} \right| \le \eps/P(\rb_e).$$
Since $P(\rb_e) \ge 1/(48n)$ by Lemma~\ref{lem:rbelower}, we get that
$$\left|x^\intercal \cdot B_j - x^\intercal \cdot B_{j'} \right| \le 48 \cdot n \cdot \eps,$$
as required.
\end{proof}

The next lemma states that the column player's strategy $y$ satisfies the
requirements of an $O(n^2 \cdot \eps)$-WSNE in $(A, B)$.

\begin{lemma}
\label{lem:scscol}
If $4 \le K \le 8$ and $\eps < 1/(240 \cdot n^2)$, then
$\left|(A \cdot y)_i - (A \cdot y)_{i'} \right| \le O(n^2 \cdot \eps).$
for all rows $i,i'$ for which $x_i > 0$ and $x_{i'} > 0$.
\end{lemma}
\begin{proof}
Let $C = \{ 3j \; : \; j < k \text{ and } P(\cb_1^j) \ge 11 \eps \} \cup \{3k +
j \; : \; j \ge k \text{ and } y^*_{3k + j} > 0\}$ be the set of columns that
define the non-zero entries of $y$.
We first show that for each row $i$ in $\rb_e$, we have
$(A' \cdot y^*)_i = P(C) \cdot (A \cdot y)_i \pm 42 \cdot n^2 \cdot \eps$.

For each column $j \le k$, Lemmas~\ref{lem:scs1} and~\ref{lem:colprob} imply
that $y^*_{3j} = 1/3 \cdot P(\cb_1^j) \pm \eps$ and $y^*_{3j+1} = 2/3 \cdot
P(\cb_1^j) \pm \eps$. Therefore, it can be verified by inspecting the
construction of the $\encode$ function that for each row $i \in \rb_e$, the
contribution of columns $3j$ and $3j+1$ to the payoff of the row is $y^*_{3j}
\cdot A_{i,j} \pm 2 \eps$. Also, for every column $j \ge k$, the contribution
of column $j$ to row $i$'s payoff is exactly $y^*_{3k + j} \cdot A_{i,j}$. All
other columns of row $i$ give zero payoff to the row player. So, if we consider
a strategy $\hat{y}$ that is constructed in the same way as $y$, but in which
no column is zeroed due to the $P(\cb_1^j) \ge 5\eps$ constraint, then we have
that $(A' \cdot y^*)_i = P(C) \cdot (A \cdot \hat{y})_i \pm 2 \cdot n \cdot
\eps$.

Next we account for the columns that are zeroed when we move from $\hat{y}$ to
$y$ itself. Here, any such column $j$ has probability at most $11 \cdot \eps$,
Since there are at most $n$ such columns, and since
all payoff entires are bounded by $K \le 8$, the maximum amount of payoff that we can
gain or lose by zeroing these columns is at most $88 \cdot n \cdot \eps$.
From this, we can conclude that
$(A' \cdot y^*)_i = P(C) \cdot (A \cdot y)_i \pm 90 \cdot n \cdot
\eps$.

Since $(x^*, y^*)$ is an $\eps$-WSNE of $(A', B')$, we have that
$\left| (A' \cdot y^*)_i - (A' \cdot y^*)_{i'} \right| \le \eps$ for any pair of
rows $i, i'$ for which $x^*_{i} > 0$ and $x^*_{i'} > 0$.
From this we obtain
$$ \left| P(C) \cdot (A \cdot y)_i - P(C) \cdot (A \cdot y)_{i'}
\right| \le 181 \cdot n \eps,$$
which implies
$$ \left| (A \cdot y)_i - (A \cdot y)_{i'} \right| \le (1/P(C)) \cdot 181 \cdot n \cdot \eps,$$
and since Lemma~\ref{lem:clowerfinal} implies that $P(C) \ge 11/(240n)$ we can
conclude that
$$ \left| (A \cdot y)_i - (A \cdot y)_{i'} \right| \le O(n^2 \cdot \eps),$$
as required.
\end{proof}

The combination of Lemmas~\ref{lem:scsrow} and~\ref{lem:scscol} then implies that
$(x, y)$ is a $O(n^2 \cdot \eps)$-WSNE in $(A, B)$, which concludes the proof
of Theorem~\ref{thm:scs1}.

\section{Proof of Theorem~\ref{thm:scs2}}
\label{app:scs2}

The proof of Theorem~\ref{thm:scs2} is essentially the same as the proof of
Theorem~\ref{thm:scs1}, with only a few minor modifications needed to account
for the new definitions of $S$, $T$, and $\encode$. We shall spend our time
here pointing out those differences, while for the main part referring back to
the proof of Theorem~\ref{thm:scs1}.

\begin{lemma}
\label{lem:scstwo1}
If $\eps < 1/(4n)$, then for each $j < k$, if $P(\cb_1^j) > 0$ holds in $(x^*,
y^*)$, then $P(\rb^j) > 0$ also holds.
\end{lemma}
\begin{proof}
The proof of this lemma is the same as the proof of Lemma~\ref{lem:scs1}, with
the constraint $\eps < 1/(3n)$ having to be adjusted to $\eps < 1/(4n)$ to
account for the larger size of $S$ and $T$.
\end{proof}

The following two lemmas depend on the specific payoff entries in the matrices
$S$ and $T$, as well as the payoffs generated by the $\encode$ function, so we
do need to reprove these.

\begin{lemma}
\label{lem:colprob2}
If $\eps < 1/(4n)$, and if $P(\rb_j) > 0$, then we have $y^*_{4j}, y^*_{4j+1},
y^*_{4j+2} \in [1/3 \cdot P(\cb_1^j) - \eps, 1/3 \cdot P(\cb_1^j) + 2\eps]$.
\end{lemma}
\begin{proof}
Suppose for the sake of contradiction that $y^*_{4j} < 1/3 \cdot P(\cb^j_1) - \eps$.
Then the payoff to the row player of row $3j$ is at most
$1/3 \cdot P(\cb^j_1) - \eps + y^*_{4j + 3}$.
On the other hand, there must exist an index $j' \in \{1, 2\}$ such that
\begin{align*}
y^*_{4j + j'} &\ge 1/2 \cdot (2/3 \cdot P(\cb^j_1) + \eps) \\
&= 1/3 \cdot P(\cb^j_1) + \eps/2.
\end{align*}
Therefore the payoff of row $3j + j'$ is at least
$1/3 \cdot P(\cb^j_1) + \eps/2 + y^*_{4j + 3}$. So the difference in payoff
between row $3j + j'$ and row $3j$ is at least
\begin{align*}
1/3 \cdot P(\cb^j_1) + \eps/2 + y^*_{4j + 3} -
1/3 \cdot P(\cb^j_1) + \eps - y^*_{4j + 3} &= 3\eps/2,
\end{align*}
meaning that the row player cannot play row $3j$ in an $\eps$-WSNE. But this
means that the column player must receive payoff $0$ for column $4j + 1$. Since
every row of $B'$ contains at least one payoff that is greater than or equal to
$1$, there must exist a column that has payoff at least $1/(4n)$. Since $\eps <
1/(4n)$, we can conclude that the column player places zero probability on
column $4j + 1$.

This now means that $y^*_{4j + 2} \ge 2/3 \cdot P(\cb^j_1) - \eps$, and using a
similar argument to the one above, we can conclude that the row player cannot
play row $3j + 1$, so the row player must only play row $3j + 2$. But now,
again using a similar argument to the one above regarding the column player, we
can conclude that the column player cannot play column $4j + 2$, meaning that
the column player places probability  $P(\cb^j_1)$ on column $4j$,
contradicting our original assumption.

We can prove that $y^*_{4j+1} \ge 1/3 \cdot P(\cb^j_1) - \eps$ and $y^*_{4j+2}
\ge 1/3 \cdot P(\cb^j_1) - \eps$ both hold using identical techniques.
Then the fact that each of the columns in $\cb^j_1$ are assigned at most
$1/3 \cdot P(\cb_1^j) + 2\eps$ probability can be shown using the fact that $\sum_{j \in
\cb_1^j} c_j = P(\cb_1^j)$.
\end{proof}

\begin{lemma}
\label{lem:cb2lower2}
If $P(\rb^j) > 0$ and $P(\cb^j_1) \ge 18 \eps$ hold in $(x^*, y^*)$ for
some column $j < k$, and if $\eps < 1/(4n)$, then $P(\cb^j_2) > 0$ also holds.
\end{lemma}
\begin{proof}
Let us assume for the sake of contradiction that
$P(\cb^j_2) = 0$.
We can invoke Lemma~\ref{lem:colprob2} to conclude that $y^*_{4j}$, $y^*_{4j +
1}$ and $y^*_{4j + 2}$ are all upper bounded by $1/3 \cdot
P(\cb_1^j) + 2\eps$.
Since $P(\cb^j_2) = 0$ by assumption, the payoff of rows $3j$, $3j + 1$, and
$3j + 2$ can therefore be upper bounded by $1/3 \cdot P(\cb_1^j) + 2\eps$.

On the other hand, observe that by construction the $\encode$ function always
produces some row $i \in \rb_e$ that contains a payoff of $1$ in columns $4j +
1$ and $4j + 2$. Lemma~\ref{lem:colprob2} implies that both of these columns
are assigned at least
$1/3 \cdot P(\cb_1^j) - \eps$ probability, so the payoff of row $i$ is at
least
$2/3 \cdot P(\cb_1^j) - 2\eps$.

So the difference in payoff between row $i$ and either of the rows $3j$ or $3j
+ 1$ is at least
\begin{align*}
& 2/3 \cdot P(\cb_1^j) - 2\eps - 1/3 \cdot P(\cb_1^j) - 2\eps \\
& = 1/3 \cdot P(\cb_1^j) - 4\eps \\
& \ge 2 \eps,
\end{align*}
where in the final line we used the fact that
$P(\cb_1^j) \ge 18\eps$.
Hence, the row player cannot place probability on row $3j$, $3j + 1$, or $3j +
2$
in an $\eps$-WSNE, contradicting the assumption that $P(\rb^j) > 0$.
\end{proof}

From now on, we can follow the proof of Theorem~\ref{thm:scs1} more or less
directly, with only a few modified constants due to the differences in the
previous two lemmas.

\begin{lemma}
\label{lem:rbelowera2}
If $P(\cb_2^j) > 0$ holds for some column $j < k$ in $(x^*, y^*)$, and $\eps <
1/(8n)$, then $P(\rb_e) \ge 1/(O(n))$ also holds.
\end{lemma}
\begin{proof}
The proof is identical to the proof of Lemma~\ref{lem:rbelowera}. The
difference in the resulting bound is explained by the different bound on $\eps$
that we assume, and the fact that all payoffs in rows in $\rb_e$ are now bounded
by $1$, rather than $K = 8$ in the other proof.
\end{proof}

\begin{lemma}
\label{lem:rbelower2}
If $\eps < 1/(18n)$, then $P(\rb_e) \ge 1/(O(n))$ holds in $x^*$.
\end{lemma}
\begin{proof}
The proof here is identical to the proof of Lemma~\ref{lem:rbelower}, noting
that our tighter upper bound on $\eps$ still allows us to conclude $P(\cb_1^j) >
18\eps$ when it is necessary to invoke Lemma~\ref{lem:cb2lower2}.
\end{proof}

\begin{lemma}
\label{lem:clower2}
If $\eps < 1/(8n)$ then
$\sum_{j < k} P(\cb_1^j) + \sum_{j \ge k} y^*_{4k + j} \ge 1/(O(n))$ holds in~$y^*$.
\end{lemma}
\begin{proof}
The proof is identical to the proof of Lemma~\ref{lem:clower}, noting that in
this case all payoff entries are bounded by $1$ instead of $K/2 = 4$, which
explains the stronger lower bound that we obtain.
\end{proof}

\begin{lemma}
\label{lem:clowerfinal2}
Suppose $\eps < 1/(96 \cdot n^2)$ then
if $C = \{j < k \; : \; P(\cb_1^j) \ge 18 \eps \}$ we have
$\sum_{j \in C} P(\cb_1^j) + \sum_{j \ge k} y^*_{4k + j} \ge 1/(O(n))$ holds in~$y^*$.
\end{lemma}
\begin{proof}
This proof is identical to the proof of Lemma~\ref{lem:clowerfinal}, where only
the constants need to be changed to account for the different constants given
to us by the previous lemmas.
\end{proof}

\begin{lemma}
\label{lem:scsrow2}
If $\eps < 1/(8n)$, then
$\left|x^\intercal \cdot B_j - x^\intercal \cdot B_{j'} \right| \le O(n \cdot \eps).$
for all columns $j,j'$ for which $y_j > 0$ and $y_{j'} > 0$.
\end{lemma}
\begin{proof}
This proof is identical to the proof of Lemma~\ref{lem:scsrow}.
\end{proof}

\begin{lemma}
\label{lem:scscol2}
If $\eps < 1/(96 \cdot n^2)$, then
$\left|(A \cdot y)_i - (A \cdot y)_{i'} \right| \le O(n^2 \cdot \eps).$
for all rows $i,i'$ for which $x_i > 0$ and $x_{i'} > 0$.
\end{lemma}
\begin{proof}
This proof is identical to the proof of Lemma~\ref{lem:scscol}.
\end{proof}

This concludes the proof of Theorem~\ref{thm:scs2}.

\section{Proof of Theorem~\ref{thm:dcs}}
\label{app:dcsproof}

Throughout this proof we suppose that $(x^*, y^*)$ be an $\eps$-WSNE of $(A', B')$.
We also use the $P$ notation to refer to the amount of probability assigned to
sets of rows or columns, as defined in Appendix~\ref{app:scs1proof}.

We start by showing some relationships between the columns in $\cb_1^i$ and the
rows in $\rb_1^i$ and $\rb_2^i$.

\begin{lemma}
\label{lem:dcs1}
If $\eps < 1/(5n)$ then
for each row $i < k$, we have the following.
\begin{itemize}
\item If $c^*_{5i+1} > 0$ then $r^*_{4i} > 0$.
\item If $c^*_{5i+2} > 0$ then $r^*_{4i+2} > 0$.
\item If $c^*_{5i} > 0$ then $r^*_{4i+1} + r^{*}_{4i+3} > 0$.
\end{itemize}
\end{lemma}
\begin{proof}
We start with the first claim. Observe that column $5i+1$ contains only a
single $1$, which lies in row $4i$ and arises from the matrix $T_1$. Hence, if
we suppose for the sake of contradiction that $r^*_{4i} = 0$, then the payoff
to the column player for playing column $5i+1$ is zero. On the other hand,
observe that every row of $B$ contains at least one $1$, so there must exist a
column for which the column player receives payoff at least $1/(5n)$. Since
$\eps < 1/(5n)$, we have that the column player cannot place probability on
column $5i+1$ in an $\eps$-WSNE, which provides our contradiction.

The proofs for the other two claims can be shown using exactly the same
technique.
\end{proof}

Next we show the critical lemma for the dual column simulation: the probability
assigned to the first column in $\cb_1^i$ is equal to the sum of the
probabilities of the other two columns. This allows us to simulate the $2$
payoffs that were in the two columns with a single $2$ in the first column of
the block.

\begin{lemma}
\label{lem:dcs2}
If $\eps < 1/(5n)$, then
for each row $i < k$, if $P(\cb_1^i) > 0$ and $P(\rb_1^i) + P(\rb_2^i) > 0$,
then $$c^*_{5i} = c^*_{5i+1} + c^*_{5i+2} \pm \eps.$$
\end{lemma}
\begin{proof}
We first prove that $c^*_{5i} \le c^*_{5i+1} + c^*_{5i+2} + \eps$. Suppose for
the sake of contradiction that this is not the case. Then the row player's
payoff for playing row $4i$ is at least
\begin{equation*}
c^*_{5i} + c^*_{5i+3} > c^*_{5i+1} + c^*_{5i+2} + \eps + c^*_{5i+3},
\end{equation*}
meaning that the gap between the payoff of row $4i$ and the payoff of row $4i +
1$ is strictly more than $\eps$. Hence we must have $x^*_{4i+1} = 0$. Exactly
the same argument can also be used to prove that $x^*_{4i+3} = 0$. So by the
contrapositive of the third claim of Lemma~\ref{lem:dcs1} we have
$c^*_{5i} = 0$, which contradicts our assumption that
$c^*_{5i} > c^*_{5i+1} + c^*_{5i+2} + \eps$.

Next we prove that $c^*_{5i} \ge c^*_{5i+1} + c^*_{5i+2} - \eps$. Suppose for
the sake of contradiction that this is not the case. This means that the payoff
for playing row $4i + 1$ is at least
\begin{equation*}
c^*_{5i+1} + c^*_{5i+2} + c^*_{5i+3} > c^*_{5i} + \eps + c^*_{5i+3},
\end{equation*}
meaning that gap between the payoffs of row $4i + 1$ and row $4i$ is strictly
more than $\eps$, so we must have $x^*_{4i} = 0$. Exactly the same argument can
be used to show that $x^*_{4i+2} = 0$. But now the contrapositives of the first
two claims of Lemma~\ref{lem:dcs1} imply that $c^*_{5i+1} = 0$ and $c^*_{5i+2}
= 0$, which contradicts our assumption that $c^*_{5i} + \eps < c^*_{5i+1} +
c^*_{5i+2}$.
\end{proof}

Finally we show a relationship between the second and third columns of
$\cb_1^i$ with the blocks $\cb_2^i$ and $\cb_3^i$.

\begin{lemma}
\label{lem:dcs3}
If $\eps < 1/(5n)$, then the following statements hold.
\begin{itemize}
\item If $c^*_{5i+1} \ge 3 \eps$ then $c^*_{5i + 3} > 0$.
\item If $c^*_{5i+2} \ge 3 \eps$ then $c^*_{5i + 4} > 0$.
\end{itemize}
\end{lemma}
\begin{proof}
We begin with the first claim.
Since $c^*_{5i+1} > 0$, Lemma~\ref{lem:dcs1} implies that $r^*_{4i} > 0$. This
implies that the precondition of Lemma~\ref{lem:dcs2} holds,
so we have $c^*_{5i} = c^*_{5i+1} + c^*_{5i+2} \pm \eps$.

Now suppose for the sake of contradiction that
$c^*_{5i + 3} = 0$. This means that the payoff of row $4i$ is exactly
$1 \cdot c^*_{5i}$, but on the other hand, by the construction of the $\encode$
function, there is some row in
$\rb_e$ that has payoff at least $2 \cdot c^*_{5i}$. So the difference in payoff
between these two rows is at least
\begin{align*}
c^*_{5i} &\ge c^*_{5i+1} + c^*_{5i+2} - \eps \\
& \ge c^*_{5i+1} - \eps \\
& \ge 3\eps - \eps = 2\eps
\end{align*}
where in the final inequality we used the assumption that
$c^*_{5i+1} \ge 3 \eps$. Hence we must have $r^*_{4i} = 0$, which is a
contradiction.

The proof of the second claim can be proved in an entirely symmetric manner.
\end{proof}

We now turn our attention towards showing that the strategies $x$ and $y$ are
well defined. We first show a lower bound on the amount of probability used in
$y'$.

\begin{lemma}
\label{lem:dcsyplower}
If $\eps < 1/(50n^2)$ then
we have $\sum_j y'_j \ge 1/(100 n)$.
\end{lemma}
\begin{proof}
Suppose for the sake of contradiction that this is not the case,
meaning that we have
$\sum_j y'_j < 1/(100 n)$.

We first argue that $\sum_{i < k} P(\cb_1^i) + \sum_{j > 5k} c^*_{j} <
1/(25 n)$ should hold, noting that the only difference between this sum
and the sum in our assumption is that now we sum over the columns $5i$ for each $i < k$.
To see why this is the case, observe that
for each row
$i < k$, if $P(\cb_1^i) > 0$, then Lemma~\ref{lem:dcs1} implies that the
precondition of Lemma~\ref{lem:dcs2} holds, and so we have
$c^*_{5i} = c^*_{5i+1} + c^*_{5i+2} \pm \eps$, meaning that
$P(\cb_1^i) \le 2 \cdot c^*_{5i+1} + 2 \cdot c^*_{5i+2} + \eps$.
This means that
\begin{align*}
\sum_{i < k} P(\cb_1^i) + \sum_{j > 5k} c^*_{j} &< 1/(50 n) + n \eps \\
&< 1/(25 n),
\end{align*}
where the final inequality uses the assumption that
$\eps < 1/(50n^2)$.

This then implies that $\sum_{i < k} P(\cb_2^i) + P(\cb_3^i) \ge 1 - 1/(25n)$.
Since columns in $\cb_2^i$ and $\cb_3^i$ give payoff to the row player only in row
blocks $\rb_i^1$ and $\rb_i^2$, we get that the payoff of all rows in $\rb_e$
can be at most $2 \cdot 1/(25n) \le 1/(12n)$. On the other hand, since every
column of $A'$ contains at least one payoff entry that is greater than or equal
to $1$, there must exist a row that provides payoff at least $1/(4n)$ to the
row player. Since $\eps < 1/(6n)$, we can conclude that the row player places
zero probability on the rows in $\rb_e$.

Now note however that for each column in $\cb_2^i$ and $\cb_3^i$, the column
player only receives non-zero payoff for rows in $\rb_e$. Hence, all such
columns must receive payoff zero. On the other hand, since each row of $B'$
contains at least one payoff entry that is greater than or equal to $1$, the
column player has a column with payoff at least $1/(5n)$. Since $\eps < 1/(5n)$
we get that the column player places zero probability on columns in
$\cb_2^i$ and $\cb_3^i$ for some $i$. But this provides a contradiction with
the fact that
$\sum_{i < k} P(\cb_2^i) + P(\cb_3^i) \ge 1 - 1/(25n)$.
\end{proof}

The strategy $y''$ is produced from $y'$ by zeroing some of the columns of
$y'$. The next lemma shows that the amount of probability lost in this
operation is relatively small.

\begin{lemma}
\label{lem:dcsypplower}
If $\eps < 1/(600 n^2)$ then we have $\sum_j y''_j \ge 1/(200 n)$.
\end{lemma}
\begin{proof}
Lemma~\ref{lem:dcsyplower} implies that
$\sum_j y'_j \ge 1/(100 n)$, and $y''$ is obtained from $y'$ by
zeroing
any column $j$ with $\lfloor j/2 \rfloor < k$ and $y'_j < 3 \eps$. Clearly the
amount that is lost due to this operation is at most
$n \cdot 3 \eps$. Since $\eps < 1/(600 n^2)$ by assumption, the amount
that is lost is therefore bounded by $1/(200 n)$. Thus we have
$\sum_j y''_j \ge 1/(100 n) - 1/(200 n) = 1/(200 n)$.
\end{proof}

This implies that the re-normalization used to define $y$ cannot fail, and so
$y$ is well defined.

Next we show that the strategy $x$ is well defined, by showing a lower bound on
the amount of probability that is used in $x'$.

\begin{lemma}
\label{lem:dcsxlower}
If $\eps < 1/(100 n)$, then we have $\sum_{i} x'_i \ge 1/(100 n)$.
\end{lemma}
\begin{proof}
Since $x'$ is defined by copying the probabilities used by $x^*$ in the block
$\rb_e$, it is sufficient to prove that $\sum_{i \in \rb_e} x^*_i \ge 1/(100
n)$.

Suppose for the sake of contradiction that we have $\sum_{i \in \rb_e} x^*_i < 1/(100 n)$.
Note that for each column $j$ that lies in $\cb_2^i$ or $\cb_3^i$ for some $i$
must therefore receive payoff at most $2/(100 n) = 1/(50 n)$, because the
maximum payoff in $B$ is $2$. The same property holds
for the columns $j \ge 5k$. Since the column player always has a column with
payoff at least $1/(5n)$, and since $\eps < 1/(100 n)$, the column player must
place zero probability on all such columns.

This implies that all probability of the column player must be allocated to
column blocks of the form $\cb_1^i$. There are at most $n$ such column blocks,
so there must therefore be at least one row $i$ for which $P(\cb_1^i) \ge 1/n
> 25 \eps$.

We argue that we must have $y^*_{5i + 1} > 3 \eps$ or $y^*_{5i + 2} > 3
\eps$, because if not, then we would have $y^*_{5i} \ge 19 \eps$, meaning
that the third point of Lemma~\ref{lem:dcs1} would imply $P(\rb_1^i) > 0$ and
$P(\rb_2^i) > 0$, and then Lemma~\ref{lem:dcs2} can be invoked to argue that
$c^*_{5i} = c^*_{5i+1} + c^*_{5i+2} \pm \eps$, which contradicts our
assumptions about $c^*_{5i+1}$ and $c^*_{5i+2}$ both being small.

This now allows us to invoke one of the two claims of Lemma~\ref{lem:dcs3} to
argue that either $P(\cb_2^i) > 0$ or $
P(\cb_3^i) > 0$. But this is a contradiction, because we have already proved
that the column player may not place probability on either of those two column
blocks.
\end{proof}

Finally, we turn our attention towards the quality of the resulting WSNE. The
following lemma considers the column player.

\begin{lemma}
\label{lem:dcsywsne}
If $\eps < 1/(100 n)$ then we have $|x^\intercal B_j - x^\intercal B_{j'}| \le 100 \cdot n
\cdot \eps$ for all columns $j$ and $j'$ for which $y_j> 0$ and $y_{j'} > 0$.
\end{lemma}
\begin{proof}
Since $(x^*, y^*)$ is an $\eps$-WSNE of $(A',
B')$, for each pair of columns $j, j'$ with $y^*_j > 0$ and $y^*_{j'} > 0$ we
have that $\left|(x^* \cdot B')_j - (x^* \cdot B')_{j'} \right| \le \eps$.

For each column $j$, let $i = \lfloor j/2 \rfloor$. If $i < k$ and $j$ is
even, then
we have $y_j > 0$ only if $c^*_{5i + 1} \ge 3 \eps$, and then
Lemma~\ref{lem:dcs3} implies that
$c^*_{5i + 3} \ge 0$. Observe that the payoff of column $5i+3$ in $A'$ is
$P(\rb_e) \cdot
\cdot x^\intercal \cdot
B_j$ by construction. Exactly the same argument also works for the case where
$j$ is odd, this time using column $5i + 4$ as an intermediary.
Finally, we also have that the payoff of any column for which $i \ge k$
is given by
$P(\rb_e) \cdot \cdot x^\intercal \cdot B_j$ by construction.

So for each pair of columns $j, j'$ for which $y_j > 0$ and $y_{j'} > 0$, we have
$$\left|P(\rb_e) \cdot x^\intercal \cdot B_j - P(\rb_e) \cdot x^\intercal \cdot B_{j'} \right|
\le \eps,$$
which implies
$$\left|x^\intercal \cdot B_j - x^\intercal \cdot B_{j'} \right| \le \eps/P(\rb_e).$$
Since $P(\rb_e) \ge 1/(100n)$ by Lemma~\ref{lem:dcsxlower}, we get that
$$\left|x^\intercal \cdot B_j - x^\intercal \cdot B_{j'} \right| \le 100 \cdot n \cdot \eps,$$
as required.
\end{proof}

The next lemma likewise considers the quality of the WSNE achieved by the row
player.

\begin{lemma}
\label{lem:dcsxwsne}
If $\eps < 1/(600 n^2)$
$\left|(A \cdot y)_i - (A \cdot y)_{i'} \right| \le 1200 \cdot n^2 \cdot \eps.$
for all rows $i,i'$ for which $x_i > 0$ and $x_{i'} > 0$.
\end{lemma}
\begin{proof}
Let $C = \{5i + 1, 5i+2 \; : \; i < k\} \cup \{j \; : \; 5k \le j \le 5k + (n-1
- j)\}$ be the set of columns that are used to define $y'$.
We first argue that $( A' \cdot y^*)_j = P(C) \cdot (A \cdot y')_j \pm
2 \eps$ for all $j \in C$, where we suppose for the sake of
convenience that the columns of $A$ are indexed by the corresponding columns in
$A'$.

First consider a row $i \ge k$. Here we directly have that
$(A' \cdot y^*)_j = P(C) \cdot (A \cdot y')_j$, since the payoff of any such
row is entirely determined by the column player's strategy
over the columns in $\{j \; : \; 5k \le j \le 5k + (n-1 - j)\}$.

Next consider a row $i < k$. Here, we need to pay attention to the encoding of
the row via the $\encode$ function. In particular, since
Lemma~\ref{lem:dcs2} implies that $c^*_{5i} = c^*_{5i+1} + c^*_{5i+2} \pm
\eps$, we get that
$(A' \cdot y^*)_j = P(C) \cdot (A \cdot y')_j \pm 2 \eps$, with the $2 \eps$
error arising due to the $\eps$ error in the
probability $c^*_{5i}$ which is then multiplied by the payoff of $2$ in the row.

So we have shown that
$(A' \cdot y^*)_j = P(C) \cdot (A \cdot y')_j \pm 2 \eps$. Next we deal with
the translation from $y'$ to $y''$, which involves setting some of the columns
in $y'$ to zero. Note that any column that is zeroed must have weight at most
$3 \eps$ by the definition of $y''$. Thus at most $n \cdot 3 \cdot \eps$ weight
is lost in this way. Therefore, the payoffs of each of the rows can increase or
decrease by at most $n \cdot 3 \cdot \eps$, meaning that we have shown that
$(A' \cdot y^*)_j = P(C) \cdot (A \cdot y'')_j \pm (3n + 2) \cdot \eps$.

Since $(x^*, y^*)$ is an $\eps$-WSNE of $(A', B')$, we have that
$\left| (A' \cdot y^*)_i - (A' \cdot y^*)_{i'} \right| \le \eps$ for any pair of
rows $i, i'$ for which $x^*_{i} > 0$ and $x^*_{i'} > 0$.
From this we obtain
$$ \left| P(C) \cdot (A \cdot y)_i - P(C) \cdot (A \cdot y)_{i'}
\right| \le (3n + 3) \cdot \eps,$$
which implies
$$ \left| (A \cdot y)_i - (A \cdot y)_{i'} \right| \le (1/P(C)) \cdot (3n + 3)
\cdot \eps,$$
and since Lemma~\ref{lem:dcsyplower} implies that $P(C) \ge 1/(200n)$ we can
conclude that
\begin{align*}
\left| (A \cdot y)_i - (A \cdot y)_{i'} \right| &\le  200 \cdot n \cdot (3n + 2) \cdot
\eps \\
& \le 1200 \cdot n^2 \cdot \eps,
\end{align*}
where the final inequality uses the fact that $n \ge 1$ and therefore $(3n +
3) \le 6n$.
\end{proof}

The combination of Lemmas~\ref{lem:dcsywsne} and~\ref{lem:dcsxwsne} implies
that $(x, y)$ is a $(1200 \cdot n^2 \cdot \eps)$-WSNE of $(A, B)$, which
completes the proof of Theorem~\ref{thm:dcs}.

\section{Proofs for the Payoff Structure of Each Simulation}
\label{app:payoffs}

The following proof shows that the initial application of a first-type single
column simulation produces a stage 1 game.

\begin{proof}[Proof of Lemma~\ref{lem:scs1ap}]
We start by enumerating the columns of $A'$.
\begin{itemize}
\item For a column $3j$  for some $j < k$, ie.\ the first column in $\cb_1^j$,
we have a payoff of $2$ arising from the matrix $S$, and at most two 1s arising
from the output of the $\encode$ function (here we note that a column with
three 1s and a $K$ is not permitted by
Definition~\ref{def:resbi}).

\item For a column $3j + 1$ for some $j < k$, ie.\ the second column in
$\cb_1^j$, we have a payoff of $1$ arising from the matrix $S$, and at most two
$1$s alongside a single value of $K/2$ arising from the $\encode$ function. So
we have at most three $1$s and a $4$, which is permitted by a stage one game.

\item For a column $3j + 2$, ie.\ the only column in $\cb_2^j$, we have exactly two
$1$s from the two entries in row block $\rb_j$ and nothing else.

\item For a column $j \ge k$, we have not changed the input matrix. Any such
column did not contain a value $K$ by definition,  so we still
have at most two $1$ payoffs.
\end{itemize}
Hence, all columns of $A'$ satisfy the requirements of a stage 1 game.

We now consider the rows of the matrix $A'$.
\begin{itemize}
\item Each row $2j$ for some column $j < k$, ie.\ the first row in $\rb_j$,
contains a payoff of $2$ from $S$ and a single $1$ from $\cb_2^j$.

\item Each row $2j + 1$ for some column $j < k$, ie.\ the second row in
$\rb_j$,
contains a $1$ from $S$ and a single $1$ from $\cb_2^j$.

\item For a row $i \in \rb_e$ we have the following.
\begin{itemize}
\item If the corresponding row of $A$ contained two $K$s, then the
corresponding row of $A'$ now contains two $K/2$s, since the $\encode$ function
does not introduce any new payoffs in the row other than the $K/2$ payoffs.

\item If the corresponding row of $A$ contained at most three entries from
$\{1, 2\}$, then note that by definition, each of the $2$ payoffs must have
shared a column with a $K = 8$ payoff. Thus the encode function turned each of
those $2$ payoffs into $1$ payoffs, and so the row now contains at most three
$1$ payoffs.
\end{itemize}
\end{itemize}
Hence, we have shown that the rows of $A'$ satisfy the requirements of a stage
1 game.
\end{proof}

The next proof shows that applying a first-type single column simulation to a
stage 1 game produces a stage 2 game.

\begin{proof}[Proof of Lemma~\ref{lem:scs1app}]
Let us begin by considering the columns of $A'$.
\begin{itemize}
\item For a column $3j$  for some $j < k$, ie.\ the first column in $\cb_1^j$,
we have a payoff of $2$ arising from the matrix $S$, and at most three 1s arising
from the output of the $\encode$ function.

\item For a column $3j + 1$ for some $j < k$, ie.\ the second column in
$\cb_1^j$, we have a payoff of $1$ arising from the matrix $S$,
along with a single $K/2 = 2$ payoff.  Here we use the fact that a $4$ payoff
cannot share a column with a $2$ payoff in a stage 1 game.

\item For a column $3j + 2$, ie.\ the only column in $\cb_2^j$, we have exactly two
$1$s from the two entries in row block $\rb_j$ and nothing else.

\item For a column $j \ge k$, we have not changed the input matrix. Any such
column did not contain a value $K$ by definition,  so we still have at most
three non-zero payoffs from the set $\{1, 2\}$ with at most one $2$.
\end{itemize}
Thus $A'$ satisfies the requirements of a stage 2 game.

Now we consider the rows of $A'$.
\begin{itemize}
\item Each row $2j$ for some column $j < k$, ie.\ the first row in $\rb_j$,
contains a payoff of $2$ from $S$ and a single $1$ from $\cb_2^j$.

\item Each row $2j + 1$ for some column $j < k$, ie.\ the second row in
$\rb_j$, contains a $1$ from $S$ and a single $1$ from $\cb_2^j$.

\item For a row $i \in \rb_e$ we have the following.
\begin{itemize}
\item If the corresponding row of $A$ contained two $K$s, then the
corresponding row of $A'$ now contains two $K/2$s, since the $\encode$ function
does not introduce any new payoffs in the row other than the $K/2$ payoffs.

\item If the corresponding row of $A$ did not contains a $K$, then observe that
the payoff structure of the row cannot have changed. In particular, in a
Stage~1 game no $2$ payoff shares a column with a $4$ payoff, so any $2$ payoff
will not be encoded and will still appear in $A'$. The encode function may have
mapped each $1$ payoff into a $1$ payoff in $A'$, but the encode function does
not create or destroy $1$s, so the resulting row of $A'$ contains either at
most three $1$s, or it contains a single $2$ along with at most one $1$.
\end{itemize}
\end{itemize}
Finally, we note that for each row that contains exactly two $2$s, each of the
$2$ payoffs shares a column with exactly one $1$ payoff. This is because
a stage $1$ game never has a $2$ payoff in the same column as a $4$ payoff,
meaning that the output of the $\encode$ function always puts the $K/2 = 2$
payoff by itself in a column, and this $2$ payoff then has exactly one $1$
payoff above it arising from the matrix $S$.
\end{proof}

The following proof shows that $B$ is unaffected when we apply a first-type
single column simulation.

\begin{proof}[Proof of Lemma~\ref{lem:scs1bp}]
Observe that every column of $B'$ is either a column from $B$, or a column that
contains a single $1$ arising from the matrix $T$. Likewise, each row of $B'$
is either a row from $B$, or a row that contains a single $1$ arising from the
matrix $T$.
\end{proof}

The following proof shows that applying a dual column simulation to a stage 2
game produces a stage 3 game.

\begin{proof}[Proof of Lemma~\ref{lem:dcsapayoffs}]
Let us first enumerate the columns of $A'$.
\begin{itemize}
\item For each $i < k$, the column $5i$ contains $1$s in rows $4i$ and $4i +
2$, both arising from the matrix $S$, and a $2$ in some row in $\rb_e$ arising
from the $\encode$ function.

\item For each $i < k$, the columns $5i + 1$ and $5i + 2$ each
$1$s in rows $4i+1$ and $4i +
3$, both arising from the matrix $S$, and a $1$ in some row in $\rb_e$ arising
from the $\encode$ function.

\item For each $i < k$, the columns $5i + 3$ and $5i + 4$ each contain two $1$s
in $\rb_1^i$ and $\rb_2^i$, respectively.

\item For each $j \ge 5k$, column $j$ contains exactly the payoffs in column
$B_{j - 5k}$ of $A$, and this column contains at most four non-zero payoffs lying in the set $\{1, 2\}$ with at most one payoff of $2$.
\end{itemize}
So we have shown that all columns of $A'$ contain at most four non-zero
payoffs lying in the set $\{1, 2\}$ with at most one payoff of $2$.

Next we consider the rows of $A'$.
\begin{itemize}
\item For each row $i < k$, the matrix $A$ contained exactly two $2$ payoffs,
and by construction, $A'$ now contains exactly one $2$ payoff.

\item For each row $i \ge k$, since the $\encode$ function does not introduce or
destroy a 1 payoff, the matrix $A'$ contains exactly the same
payoffs as the corresponding row of $A$, so the row
either contains at most three $1$s, or a $2$ and at most one $1$.
\end{itemize}
Hence we have shown that each row of $A'$ satisfies the required properties.
\end{proof}

The following proof states that $B$ is unaffected when we apply a dual column
simulation.

\begin{proof}[Proof of Lemma~\ref{lem:dcsbpayoffs}]
Observe that each row and of $B'$ is either a row of $B$ with extra zeroes
added, or contains exactly one $1$ arising from either $T_1$ or $T_2$. Likewise, each
column of $B'$ is either a column of $B$ with extra zeroes added, or contains
exactly two $1$s arising from $T_1$ and $T_2$.
\end{proof}

The following proof shows that applying a second-type single column simulation
to a stage 3 game gives us a $3$-sparse win-lose game.

\begin{proof}[Proof of Lemma~\ref{lem:scs2bp}]
Let us enumerate the columns of $A'$.
\begin{itemize}
\item Columns of the form $4j$ and $4j + 1$ with $j < k$ have a single $1$ from the matrix
$S$ and at most two $1$s from the output of the $\encode$ function.

\item Columns of the form $4j + 2$ with $j < k$ have a single $1$
from the matrix $S$ and a single $1$ from the output of the $\encode$ function.

\item Columns of the form $4j + 3$ have exactly three $1$s in the block
$\rb^j$.

\item Columns $j \ge k$ have at most three $1$s by assumption.
\end{itemize}
Now we inspect the rows of $A'$.
\begin{itemize}
\item Rows in $\rb^j$ for some $j$ have a single $1$ from the matrix $S$, and a
single $1$ in $\cb_2^j$.

\item For rows in $\rb_e$, note that the $\encode$ function does not create or
destroy a $1$, but it does split a $2$ into two $1$s. However, if a row has a
$2$ in a stage 3 game, then that row contains at most one $1$, hence the
resulting row of $A'$ contains at most three $1$s.
\end{itemize}
\end{proof}

The following proof shows that the matrix $B$ is unaffected by a second-type
single column simulation.

\begin{proof}[Proof of Lemma~\ref{lem:scs2ap}]
Observe that every column of $B'$ is either a column from $B$, or a column that
contains a single $1$ arising from the matrix $T$. Likewise, each row of $B'$
is either a row from $B$, or a row that contains a single $1$ arising from the
matrix $T$.
\end{proof}

\end{document}